\documentclass[lettersize,journal]{IEEEtran}
\usepackage{algorithmic}
\usepackage{xcolor}
\def\BibTeX{{\rm B\kern-.05em{\sc i\kern-.025em b}\kern-.08em
    T\kern-.1667em\lower.7ex\hbox{E}\kern-.125emX}}

\usepackage{amsbsy,amsmath,amssymb,amsfonts,amsthm}
\usepackage{bm}
\usepackage{gensymb}
\usepackage{graphicx,wrapfig}
\usepackage{multirow}
\usepackage{booktabs}
\usepackage{mathrsfs}
\usepackage{color, soul}
\usepackage[hidelinks]{hyperref}
\usepackage{cite}
\usepackage[shortlabels]{enumitem}
\usepackage{epsfig}
\usepackage{epigraph}
\usepackage{csquotes}
\usepackage[strict]{changepage}
\usepackage{datetime2}
\usepackage{siunitx}
\usepackage{tabularx}
\usepackage{acronym}
\usepackage{orcidlink}

\usepackage{import}
\usepackage{xifthen}
\usepackage{pdfpages}
\usepackage{transparent}
\usepackage{pstricks}
\usepackage[inkscapearea=page]{svg}
\svgpath{{./graphics/}}

\usepackage{etoolbox} 
\usepackage[scale=.75]{miama}
\usepackage[subdued, LGRgreek, italic]{mathastext}
\MTgreekfont{fmm}\Mathastext[greekfont]

\usepackage{float}
\usepackage{tikz}
\usepackage{pgfplots}
\DeclareUnicodeCharacter{2212}{−}
\usepgfplotslibrary{groupplots,dateplot}
\usetikzlibrary{patterns,shapes.arrows,shapes.geometric,arrows,pgfplots.groupplots,positioning,arrows.meta}
\pgfplotsset{compat=newest}

\renewcommand{\rm}[1]{\mathrm{#1}} 
\renewcommand{\bf}[1]{\mathbf{#1}} 
\newcommand{\uli}[1]{\underline{#1}} 
\newcommand{\oli}[1]{\overline{#1}} 
\newcommand{\uarrow}{\,\uparrow} 
\newcommand{\darrow}{\,\downarrow} 
\newcommand{\nSet}{\mathcal{N}} 
\newcommand{\gSet}{\mathcal{G}} 
\newcommand{\dSet}{\mathcal{D}} 
\newcommand{\eSet}{\mathcal{E}} 
\newcommand{\tSet}{\mathcal{T}} 
\newcommand{\mSet}{\mathcal{M}} 
\newcommand{\CE}{{\scriptscriptstyle\mathrm{CE}}} 
\newcommand{\sz}{{\scriptscriptstyle 0}} 
\newcommand{\KCL}{{\scriptscriptstyle\rm{KCL}}} 
\newcommand{\KVL}{{\scriptscriptstyle\rm{KVL}}} 
\newcommand{\PR}{{\scriptscriptstyle\rm{PR}}} 
\newcommand{\SD}{{\scriptscriptstyle\rm{SD}}} 
\newcommand{\SU}{{\scriptscriptstyle\rm{SU}}} 
\newcommand{\params}{\text{\MTversion{greekfont}$\Omega$}} 
\robustify{\params}
\newcommand{\hatparams}{\params_{\CE}} 
\newcommand{\tildeparams}{\params_{\sz}} 
\newcommand{\gAsk}{g'} 
\newcommand{\bfp}{{\mathbf{p}}} 
\renewcommand{\ac}{\bm{\alpha}} 
\newcommand{\cSet}{\mathcal{C}} 
\newcommand{\intSet}{\mathcal{U}} 
\newcommand{\dataSet}{\mathcal{S}} 
\newcommand{\gen}{{\scriptscriptstyle\mathrm{G}}} 
\newcommand{\dem}{{\scriptscriptstyle\mathrm{D}}} 
\renewcommand{\line}{{\scriptscriptstyle\mathrm{L}}} 
\newcommand{\norm}[1]{\left\lVert#1\right\rVert}
\DeclareMathOperator*{\argmin}{arg\,min}
\DeclareMathOperator*{\argmax}{arg\,max}

\newcommand{\tp}[1]{{#1}^{\mathsf{T}}}
\newcommand{\knn}{{k\rm{NN1}}}
\newcommand{\reals}{{\mathbb{R}}}

\newcommand{\pb}[4][black]{%
  \parbox{#3cm}{%
    \ifthenelse{\equal{#2}{l}}{\raggedright}{%
    \ifthenelse{\equal{#2}{c}}{\centering}{%
    \ifthenelse{\equal{#2}{r}}{\raggedleft}{}}}%
    \textcolor{#1}{#4}%
  }%
}

\newtheorem{assumption}{Assumption}
\newtheorem{proposition}{Proposition}

\newcolumntype{H}{>{\setbox0=\hbox\bgroup}c<{\egroup}@{}}

\usepackage{hyperref}

\begin{document}

\title{Counterfactual Explanations for Power System Optimisation}

\author{Benjamin Fritz\,\orcidlink{0009-0006-5326-8299},~\IEEEmembership{Student Member,~IEEE,} Waqquas Bukhsh\,\orcidlink{0000-0002-5765-0747},~\IEEEmembership{Senior Member,~IEEE}\vspace{-6mm}}



\maketitle

\begin{abstract}
  Enhanced computational capabilities of modern decision-making software have allowed us to solve increasingly sophisticated optimisation problems. But in complex socio-economic, technical environments such as electricity markets, transparent operation is key to ensure a fair treatment of all parties involved, particularly regarding dispatch decisions. We address this issue by building on the concept of counterfactual explanations, answering questions such as ``\textit{Why was this generator not dispatched?}'' by identifying minimum changes in the input parameters that would have changed the optimal solution. Both DC Optimal Power Flow and Unit Commitment problems are considered, wherein the variable parameters are the spatial and temporal demand profiles, respectively. The thereby obtained explanations allow users to identify the most important differences between the real and expected market outcomes and observe which constraints have led to the solution. The framework uses a bilevel optimisation problem to find the counterfactual demand scenarios. State-of-the-art methods are compared with data-driven heuristics on the basis of computational efficiency and explanation accuracy. Results show that leveraging historical data from previously solved instances can provide significant speed benefits and allows us to derive explanations in cases where conventional methods would not be tractable. 
\end{abstract}


\section{Introduction}\label{sec:intro}


\IEEEPARstart{O}{ptimisation} software is at the core of solving many operational, planning and maintenance problems arising in modern power systems, where they introduce improved capabilities of enhanced decision-making and thereby facilitate significant economic savings.
But while computational advances have profoundly improved the speed of state-of-the-art solvers, and have expanded their applications in power systems, much less attention has been paid to \textit{explaining} the outputs of such tools. 

This becomes particularly crucial in the face of the ongoing democratisation and decentralisation of electricity systems, making them inherently more complex to interact with. In
Great Britain (GB), the system operator now accepts bids from assets with capacities below 1~MW in the \textit{Balancing Mechanism} (BM)~\cite{nationalenergysystemoperatorIncreasingFlexibilitySmallscale2024}, aiming to boost flexibility and cost-efficiency by leveraging the growing share of distributed energy resources (DER). From an optimisation perspective, this adds decision variables and expands the dimensionality of the decision space, an effect further magnified by the increasing coupling of power, heat and transport sectors.
Meanwhile, advances in solver efficiency have enabled harder problem formulations involving integer variables, non-linear constraints and multiple objectives, adding further complexity.

Despite offering improved modelling capabilities, such sophisticated models undermine a crucial aspect of multi-stakeholder decision-making: the explainability of black-box optimisation outputs. Explainability in machine learning (ML) has gained increasing attention with the rise of explainable artificial intelligence (XAI). 
Most of the works in the field of \textit{Explainable Optimisation} have emerged only very recently, mostly focusing on interpretability, i.e., deriving an easily comprehensible set of rules mapping instances to solutions. Tree-based methodologies (e.g., \cite{bertsimasVoiceOptimization2020,goerigkFrameworkInherentlyInterpretable2023}) build on the notion that decision-trees as inherently interpretable models. In this vein -- but with direct application to power system optimisation -- references \cite{stratigakosInterpretableMachineLearning2024} and \cite{lumbrerasExplainingSolutionsUnit2024} employ decision trees to introduce interpretability to the DC Optimal Power Flow and Unit Commitment problems, respectively.
While data-driven interpretability techniques can provide useful insights, they are unlikely to fully replace conventional optimisation tools in power systems. Instead, \emph{explainability} methods aim to reveal the reasons for individual model outcomes, thereby allowing the justification of single optimal decisions. In GB, the debate on ``\textit{skip-rates}'' has increased pressure on the operator to improve transparency in the BM~\cite{lcpdeltaBMSkipRates2024}. A ``\textit{skip}'' occurs when a seemingly uneconomical balancing action is chosen over an in-merit alternative. Stakeholders requesting reasons for such decisions would thus directly benefit from improved explainability.

Overall, we find that there is a gap of approaches aiming to explain individual solutions to power system optimisation problems.
We address this by proposing a framework based on counterfactual explanations (CEs) that builds on preliminary findings from a previous study \cite{fritzExplainableDCOptimal2025}, but we extend it to the UC problem, generalise the methodology, and derive supporting theoretical principles. The remainder of this paper is organised as follows. Section \ref{sec:XOpt} introduces the concept CEs, followed by a discussion of how CEs can be employed in a power system context in Section \ref{sec:CE4OPF}. We present two illustrative case studies for DCOPF and UC instances respectively, exemplifying the functioning of our approach.
Section \ref{sec:bilo} defines methods to efficiently compute CEs from bilevel problems. The experimental setup used to evaluate the performance of the proposed methodology and the accompanying results are presented in Section \ref{sec:res}. Section \ref{sec:conclusion} concludes the paper.

\subsection*{Notation}
Bold lowercase letters (e.g., $\mathbf{x}$) represent vectors, bold uppercase letters (e.g., $\mathbf{X}$) represent matrices. A transposed vector is denoted by $\bf{x}^{\mathsf{T}}$. Ordinary letters, either lowercase or uppercase, (e.g., $n$ or $N$) refer to scalars and functions. Sets are denoted by calligraphic letters, e.g. $\mathcal{X}$ and their cardinality is given by $|\mathcal{X}|$. $\emptyset$ is the empty set. A vector $\mathbf{x}=\{0,1\}^N$ is a vector of length $N$ of zeros and ones. $[X]$ is shorthand for $\{1,\ldots,X\}$. Finally, $\left\{\bf{x},\bf{y}\right\}_{i=1}^N$ is a dataset $\left\{\left(\bf{x}^{[1]},\bf{y}^{[1]}\right),\ldots,\left(\bf{x}^{[N]},\bf{y}^{[N]}\right)\right\}$.

\vspace{0mm}

\section{Counterfactual Explanations for Optimisation}\label{sec:XOpt}
Counterfactual explanations (CEs) fall into the realm of local, example-based explainability. Here, ``\emph{local}'' refers to the idea that CEs explain a single instance of the model's output with respect to (w.r.t.) a given input -- as opposed to \emph{global} explanations, which aim to characterise the overall decision behaviour across the entire input space.
In this way, CEs help users to identify cause-effect relations between the model inputs and outputs, allow exploring desired ``what-if'' scenarios, and are advocated for by studies in social and cognitive sciences 
\cite{guidottiCounterfactualExplanationsHow2024}. 

\subsection{General Concept}

While the concept of CEs can be applied analogously for ML models, this paper studies their application in the area of mathematical optimisation. We consider the following generic optimisation problem:
\begin{equation}\label{eq:fw_general}
    \begin{aligned}
        \mathbf{x}^*\quad\in\quad&\argmin_{\mathbf{x}} \quad & & f(\mathbf{x}, \params)\\
        &\text{s.t.} 	\quad & & g_i(\mathbf{x},\params)\geq {0}\quad,\quad\forall i\in\cSet\quad,
    \end{aligned}
\end{equation}
where the solution $\mathbf{x}^*$ is the minimiser of an objective function $f$ and $\params$ is a set of constant input parameters (e.g., for a linear optimisation problem this set may be $\params=\{\mathbf{A},\mathbf{b},\mathbf{c}\}$). The problem is subject to a set of constraints ${g}_i$ collected in the set $\mathcal{C}$. 
Now consider a situation in which an optimal solution $\mathbf{x}_{\sz}^*$ deviates from the user's intuitive expectation of what a typical solution for a given parameter set $\params_{\sz}$ should look like. A CE then answers the following question:
    ``\textit{How can the input parameters $\params$ be minimally perturbed, such that the new solution falls into the region in which the user expected it to be?}''
In revealing this information, the CE helps the user to better understand the impact of the parameters on the solution and identifies the key parameter changes that characterise the difference between the expectation and the actual decision.

\subsection{Bilevel Problem Formulation}

The nearest counterfactual explanation can be found from a bilevel optimisation problem which is solved after the result of problem \eqref{eq:fw_general} is available. 
Here, the goal is to find the parameter set $\hatparams$ minimally perturbed from $\tildeparams$, subject to the user's definition of a \textit{solution region} ${\mathcal{X}}$ in which they would have expected to find the solution (this may even be a single point). 
Note that, generally, ${\mathcal{X}}\cap\mathbf{x}_{\sz}^*=\emptyset$, since otherwise the trivial solution $\hatparams=\tildeparams$ would be obtained. The general CE problem then reads
\begin{subequations}\label{eq:ce_general}
    \begin{alignat}{3}
        \mkern-10mu(\hatparams\,,\,\mathbf{x}_{\CE}^*)\mkern5mu\in\mkern5mu&\argmin_{\params\,, \,\mathbf{x}} \mkern5mu & & \norm{\params-\tildeparams}_1\label{seq:ce_general_obj}\\
        & \mkern5mu\text{s.t.} &  & \mathbf{x} \in \argmin_{\mathbf{y}\,:\,{g}_i(\mathbf{y},\params)\geq 0,\forall i\in\cSet} \mkern5mu f(\mathbf{y}, \params)\label{seq:ce_general_opt_constr}\\
        & & & \mathbf{x} \in {\mathcal{X}}\label{seq:ce_general_sol_req}\\
        & & & \params\in\mathcal{M}_{\params}\mkern5mu\label{seq:ce_general_mutable},
    \end{alignat}
\end{subequations}
where $\norm{\cdot}_1$ is the $\ell_1$ distance measurement (which, with slight abuse of notation, is here assumed to be taken element-wise for each entry of each subset within $\params$). In theory, other distances measurements can be employed at this stage. We opt for the $\ell_1$-norm due to its ability to maintain a linear objective term and promote sparse explanations. 
Constraint \eqref{seq:ce_general_mutable} is added to limit the range in which the parameters $\params$ are allowed to vary. By defining the \textit{mutable parameter set} $\mathcal{M}_{\params}$, we may limit the number or nature of mutable input parameters. 


\section{Explaining Generator Dispatch Decisions}\label{sec:CE4OPF}

Thus far only a generic optimisation of shape \eqref{eq:fw_general} was considered. The following sections introduce how to generate CEs for solutions to two well-studied power system optimisation problems -- the DC Optimal Power Flow (DCOPF) and the Unit Commitment (UC) problem. Arguably, both of these problems lead to counterintuitive market outcomes (e.g., a more expensive generator was chosen over a cheaper unit), and thus motivate the usage of CEs. 

\subsection{DC Optimal Power Flow Formulation}\label{ssec:dcopf}
We consider a transmission network where $\nSet$, $\eSet$, $\gSet$ and $\dSet$ denote the sets for nodes, lines between nodes, generators and demands, respectively. The corresponding cost-minimising DCOPF problem is given as
\vspace{-0mm}
\begin{subequations}\label{eq:dcopf}
    \begin{alignat}{3}
        &\min_{\mathbf{p}^{\gen}}	& & \quad\tp{\mathbf{c}}\bfp^{\gen}	&\label{seq:dcopf_obj}\\
        &\text{s.t.} 		& 			& \mkern-9mu\sum_{\{n,m\}\in\eSet_n} p^{\line}_{n,m} = \sum_{g\in\gSet_n}p^{\gen}_g - \sum_{d\in\dSet_n}p^{\dem}_d \mkern5mu,\mkern5mu\forall n \in \nSet \quad(\bm{\lambda}^{\KCL})& \label{seq:dcopf_kcl_constr}\\
        & 					& 			& p^{\line}_{n,m} = (\delta_n-\delta_m) b_{n,m} \mkern5mu,\mkern5mu\forall \{n,m\}\in\eSet \quad(\bm{\lambda}^{\KVL})&\label{seq:dcopf_kvl_constr}\\
        & 					& 			& -\oli{\bfp}^{\line} \leq \bfp^{\line} \leq \oli{\bfp}^{\line} \quad(\uli{\bm{\lambda}}^{\line},\oli{\bm{\lambda}}^{\line})& \label{seq:dcopf_f-line_constr}\\
        & 					& 			& \uli{\bfp}^{\gen} \leq \bfp^{\gen} \leq \oli{\bfp}^{\gen} \quad(\uli{\bm{\lambda}}^{\gen},\oli{\bm{\lambda}}^{\gen})\quad,& \label{seq:dcopf_p-gen_constr}
    \end{alignat}
\end{subequations}
where
$\mathbf{c}\in\reals_{\geq0}^{|\gSet|}$ is the incremental generation cost, $\bfp^{\gen}\in\mathbb{R}_{\geq0}^{|\gSet|}$ and $\bfp^{\dem}\in\mathbb{R}_{\geq0}^{|\dSet|}$ are the active power injected/consumed by generators and demands respectively and $\bfp^{\line}\in\reals^{|\eSet|}$ contains the active power flows on all the lines in the network which are bounded by upper thermal limits $\oli{\bfp}^{\line}\in\reals^{|\eSet|}_{\geq0}$. $b_{n,m}\in\reals$ is the susceptance of the line connecting nodes $n$ and $m$, $\bm{\delta}\in\reals^{|\nSet|}$ is the vector of nodal voltage angles and $\uli{\bfp}^{\gen}$ and $\oli{\bfp}^{\gen}$ denote the lower and upper generator limits.
The associated dual variables are shown in brackets besides the respective constraints.

\subsection{Unit Commitment Formulation}\label{ssec:uc}
The UC formulation implemented in this paper implements a reduced version of \cite{morales-espanaTightCompactMILP2013}.
We aim to minimise the total operating cost over a time horizon $\tSet$:
\begin{subequations}
    \begin{alignat}{3}
    \min_{\bf{P}^{\gen}, \bf{U}}	& \, & \sum_{t\in\tSet}\Bigl[\sum_{g\in\gSet}\underbrace{c_{g1}^{\PR}p^{\gen}_{g,t} + c_{g0}^{\PR}u_{g,t}}_{\rm{(i)}} + \underbrace{c_g^{\SU}u_{g,t}^{\SU}}_{\rm{(ii)}} + \underbrace{c_g^{\SD}u_{g,t}^{\SD}}_{\rm{(iii)}} \Bigr].
    &\label{seq:uc_obj}
    \end{alignat}
	Here, (i) is the production cost, while (ii) and (iii) are the start-up and shut-down costs of unit $g$, respectively.
	The continuous dispatch variables $p_{g,t}^{\gen}\in\bf{P}^{\gen}\in\mathbb{R}^{|\gSet|\times|\tSet|}$ and the binary commitment variables $u_{g,t}\in\bf{U}\in\{0,1\}^{|\gSet|\times|\tSet|}$ span the decision space. The start-up and shut-down variables $u_{g,t}^{\SU}, u_{g,t}^{\SD}$ are directly implied by the choice of $\bf{U}$, see constraint \eqref{seq:uc_logic_constr}. For later usage we define $\bfp^{\dem}=\bigl[\sum_{d\in\dSet}p^{\dem}_{d,1},\ldots,\sum_{d\in\dSet}p^{\dem}_{d,|\tSet|}\bigr]$. Thus note that while $\bfp^{\dem}$ spanned along the nodal dimension for the DCOPF, it here stands for the net demand profile in every time step.

	The following constraints are imposed onto the problem.

    \emph{Power balance constraint}: Since transmission constraints are neglected in the UC formulation used in this work, individual nodal power balances do not need to be satisfied. Merely the overall system power needs to match.
    \begin{equation}
        \sum_{g\in\gSet}p^{\gen}_{g,t} = \sum_{d\in\dSet}p^{\dem}_{d,t}\,\left(=P^{\dem}_t\right)\qquad\forall t\in\tSet \label{seq:uc_p_balance_constr}
    \end{equation}

    \emph{Generator power limits}: Enforce on/off coupling between $u_{g,t}$ and $p_{g,t}$, and on-state generator limits.
    \begin{equation}
        u_{g,t}\uli{p}^{\gen}_g \leq p^{\gen}_{g,t} \leq u_{g,t}\oli{p}^{\gen}_g \qquad\forall t\in\tSet,\,\,\forall g\in\gSet \label{seq:uc_pG_min_max}
    \end{equation}

    \textit{Ramping limits}: The generator ramp rates are bounded by upper $R_{g}^{\uarrow}\in\mathbb{R}_{\geq0}$ and lower $R_{g}^{\darrow}\in\mathbb{R}_{\geq0}$ ramping limits.
    \begin{equation}\label{seq:uc_ramping}
        -R_{g}^{\darrow} \leq p_{g,t}^{\gen}-p_{g,t-1}^{\gen} \leq R_{g}^{\uarrow} \qquad\forall t\in\tSet,\,\,\forall g\in\gSet
    \end{equation}

	\textit{Logic constraint}: Ensures the assignment of the correct binary variables to $u_{g,t}^{\SU}$ and $u_{g,t}^{\SD}$.
    \begin{equation}
        u_{g,t}-u_{g,t-1} = u_{g,t}^{\SU} - u_{g,t}^{\SD} \qquad\forall t\in\tSet,\,\,\forall g\in\gSet\label{seq:uc_logic_constr}
    \end{equation}

	\textit{Minimum up- and downtime}: The minimum number of periods that the unit must be online $T^{\rm{U}}_g\in\mathbb{R}_{\geq0}$ and offline $T^{\rm{D}}_g\in\mathbb{R}_{\geq0}$ are guaranteed by the constraint below. The constraint further removes the possibility of a unit simultaneously starting up and shutting down as it dominates over the less strict constraint $u_{g,t}^{\SU}+u_{g,t}^{\SD}\leq1$. 
    \begin{equation}
        \sum_{\tau=t-T^{\rm{U}}_g+1}^t u_{g,\tau}^{\SU} \leq u_{g,t} \qquad\forall t\in\tSet,\,\,\forall g\in\gSet\label{seq:min_up_constr}
    \end{equation}\\[-6mm]
    \begin{equation}
        \sum_{\tau=t-T^{\rm{D}}_g+1}^t u_{g,\tau}^{\SD} \leq 1 - u_{g,t} \qquad\forall t\in\tSet,\,\,\forall g\in\gSet\label{seq:min_down_constr}
    \end{equation}
\end{subequations}

\subsection{Which question are we trying to find an explanation to?}

The definition of the solution region $\mathcal{X}$ depends on the nature of the question being asked. In power systems, such questions could be posed by the market regulator, asking the operator why a certain dispatch decision was made. In GB, this situation has led the operator to maintain a regularly updated \textit{Dispatch Transparency} dataset, containing all actions taken in the BM for a given time interval, the reasons for them and alternative actions that were dismissed~\cite{nationalenergysystemoperatorDispatchTransparency2024}. 
Motivated by this insight, we make the two following assumptions:

\begin{assumption}[Dispatch-restricting solution region]\label{as:X_DCOPF_UC}
The solution region for the DCOPF is merely restricted by the following condition:
\begin{equation}\label{eq:X_DCOPF}
\mathcal{X}_{\rm{DCOPF}}:= \{p^{\gen}_{\gAsk} \geq P_{\gAsk}\}\,\cap\,\mathbb{R}_{\geq0}^{|\gSet|-1},
\end{equation}
translating to the question of why generator $\gAsk$ had not been dispatched above a specified minimum threshold $P_{\gAsk}\in\reals_{\geq0}$, while all other dispatch values are free to vary.

For the UC problem, the solution region is defined as
\begin{equation}\label{eq:X_UC}
\mathcal{X}_{\rm{UC}}:= \{u_{\gAsk\mkern-1mu t'} =1\}\,\cap\,\{0,1\}^{(|\gSet|-1)\times|\tSet|},
\end{equation}
i.e., asking why unit $\gAsk$ was not committed at all during the time period $t'$.
\end{assumption}

\begin{assumption}[Solution Region Feasibility]\label{as:sol_reg_feas}
The solution region $\mathcal{X}$ is defined such that it contains at least one point which is feasible w.r.t. the original problem \eqref{eq:fw_general}, and, more restrictively, feasible w.r.t. problem \eqref{eq:ce_general}, i.e., there exist a parameter variation for which at least one solution in $\mathcal{X}$ becomes optimal. 
\end{assumption}

If assumption \ref{as:sol_reg_feas} is dropped, one may still identify minimal parameter changes that would make problems \eqref{eq:fw_general} and \eqref{eq:ce_general} \textit{feasible} by computing a Minimal Intractable System (MIS)~\cite{chinneckFeasibilityInfeasibilityOptimization2008}. Alternative ways of defining the solution region are theoretically possible. For instance, one could ask how parameters would need to be minimally modified to lower the flow across a congested line. Given that problem \eqref{eq:fw_general} is convex, one may even formulate $\mathcal{X}$ as to impose constraints on dual variables, allowing for questions such as ``Why is the \textit{price} at this bus higher than \qty{30}{\$/MWh} today?''. 

\subsection{What should the explanation look like?}

The structure of the CE is also influenced by the user's definition of the mutable parameter set $\mSet_{\params}$. It generally makes sense to ask oneself which problem parameters introduce the variability that ultimately leads to different solutions. The DCOPF and UC problems are both used in operational contexts, and thus require to be solved multiple times a day. If we would like to know why today we obtain a different optimal solution than yesterday, the main sources of variability are demand, renewable energy availability, and generator bids. This paper focuses on the former, and hence, makes the final following assumption:
\begin{assumption}[Demand-based parameter variation]\label{as:mutable}
The mutable parameter set is defined as
\begin{equation}
	\mSet_{\params} = \left\{\params\,:\,\bfp^{\dem}\in\mathbb{R}_{\geq0}^{|\dSet|} \,\wedge\,\params\backslash\{\bfp^{\dem}\}=\params_{\sz}\backslash\{\bfp^{\dem}_{\sz}\}\right\}\,,
\end{equation}
i.e., demand parameters can vary freely, while all other parameters are fixed to their factual values in $\params_{\sz}$. The objective function \eqref{seq:ce_general_obj} thus reduces to $\norm{\bfp^{\dem}-\bfp^{\dem}_{\sz}}$.
\end{assumption}

CEs defined by other parameter variations (such as those based on generator cost bids) are theoretically possible and would -- at least for the DCOPF -- not require a solution approach different from those presented in the following section. While the network structure itself is assumed to be fixed, users modelling long-term (stochastic) power system planning problems may also be interested in varying parameters such as line characteristics or perturbing the input scenarios.

\subsection{Illustrative Examples}

\subsubsection{DC Optimal Power Flow}

The DCOPF as described in \eqref{eq:dcopf}, is a linear problem (LP). LPs are often perceived as the simplest among optimisation problems, yet we argue that the impact of DCOPF network constraints can yield unintuitive solutions -- even for small instances. For this, we study the following example, which was already introduced in \cite{fritzExplainableDCOptimal2025}. 

Consider the network shown in Figure \ref{fig:case5}, which is based on the \texttt{case5\_pjm\_\hspace{-0.2mm}\_api} found in \cite{BenchmarksOptimalPower2021}. For simplicity, assume that $(\uli{p}^{\gen},\oli{p}^{\gen})=(0,500)\, [\unit{MW}]$ are the output bounds for all generators and all lines have a limit of $\oli{p}^{\line}=\qty{240}{MW}$. The susceptance and cost values are unchanged from the original data, where the latter is given as $\mathbf{c}=(14,15,30,40,10)\,[\unit{\$/MW}]$. Running the DCOPF problem \eqref{eq:dcopf} using $\bfp^{\dem}=(300,480,140,600,20)\,[\unit{MW}]$ as the ``factual'' demand scenario, we get the optimal dispatch $\bfp^{\gen *}=(500,306.2,485.9,14.1,233.8)\,[\unit{MW}]$. Evidently, this dispatch does not follow the merit-order. The dual variables indicate that four lines are congested (as shown in the left part of Figure \ref{fig:case5}), with the maximum value $\uli{\lambda}^{\line}_{4,5} = \qty{53.2}{\$/MW}$. At this point it is not directly evident which congestion led to the out-of-merit dispatch, or even less, how the demand values impacted the optimal solution.

Since $g_5$ is the cheapest unit, however, not fully dispatched, the user asks why it was not dispatched at least, say 400\,MW. The unique, minimal demand perturbation that satisfies this condition, is described by the variations shown in the right part of Figure \ref{fig:case5}, and consists solely of an increase of $p_{d_3}^{\dem}$ from 140\,MW to 190\,MW. 
Although $g_5$ is increased to the specified 400\,MW threshold, the remaining generators need to adjust correspondingly. Interestingly, the demand variation alleviates congestion on line \{1,4\} rather than \{4,5\}, despite the latter exhibiting the highest dual value. Through this hypothetical scenario, the CE establishes a causal link between the demand at $d_3$ and the optimal dispatch of $g_5$.
\begin{figure}[t]
        \centering
        \vspace{-1mm}
		\def\svgwidth{0.9\linewidth}
		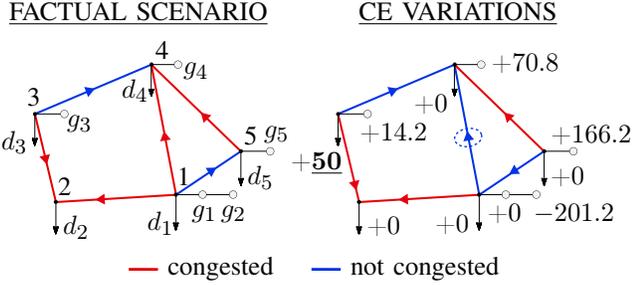
		\vspace{0mm}
		\caption{5-bus toy network (all units in MW). In the factual scenario (left), $g_5$ -- the cheapest unit -- is not fully dispatched. The user is interested in why $g_5$ was not dispatched at least 400\,MW. The counterfactual scenario (right) shows the minimum nodal demand variations that satisfy this solution requirement.}
		\label{fig:case5}
		\vspace{-5mm}
\end{figure}

\subsubsection{Unit Commitment}

For the UC problem, factors such as ramping limits, minimum up- and downtime constraints, or a unit's high minimum stable generation limit can lead to out-of-merit solutions \cite{bukhshSignificanceTimeConstraints2024}. Further, the combination of different cost factors (incremental, fixed, start-up, shut-down) may prohibit a straightforward economic comparison of the generators. 

Consider the test case defined by the generator data shown in Table \ref{tab:uc_gen_data}. Since no transmission constraints are considered, we focus solely on the aggregated demand profile $\bfp^{\dem}$ as defined in Section \ref{ssec:uc}. It is shown as a solid line in the upper plot in Figure \ref{fig:uc_example}. This factual demand scenario (labelled `0') is supplied by the optimal generator dispatch depicted by the solid lines in the lower plot. $g_1$ here takes on the role of the ``baseload unit'', delivering a large proportion of the power due to its low operating cost. 
But, due to its inflexibility, the other two units need to be activated to cover the morning and evening peaks. For scenario 0, unit $g_3$ is dispatched in the morning, while $g_2$ is used in the evening. Overall, $g_3$ has a higher operating cost than $g_2$, meaning that any substitution of $g_2$ by $g_3$, such as in the present scenario, constitutes an out-of-merit decision. The user therefore asks how the demand profile would have had to look like such that $g_2$ could have been used during the morning. For instance, the user could add the requirement $u_{g_2,8:00}=1$ to the solution region $\mathcal{X}$. 

The profiles associated with the minimally perturbed counterfactual scenario are shown as dashed lines in Figure \ref{fig:uc_example} (labelled `CE'). On the demand side, the variation is marginal -- a drop in total load of around 18\,MW at 7:00\,AM. Nevertheless, this is sufficient to ``undo'' the substitution of $g_2$ by $g_3$, resulting in a very different dispatch configuration from 5:00\,AM to 2:00\,PM. It can now be seen that the high demand ramp between 6:00\,AM and 8:00\,AM is what prohibited the activation of unit $g_2$. 
%
\begingroup
\begin{table}[t]
    \caption{Exemplary UC instance - generator data}
    \label{tab:uc_gen_data}
    \vspace{-5mm}
    \begin{center}
	\resizebox{\linewidth}{!}{
		\begin{tabular}{llllllll}
			\toprule
			& $(\uli{p}^{\gen},\oli{p}^{\gen})$ & $c_{1}^{\rm{PR}}$ & $c_{0}^{\rm{PR}}$ & $c^{\rm{SU}}$ & $c^{\rm{SD}}$ & $R^{\uarrow},R^{\darrow}$ & $T^{\rm{U}},T^{\rm{D}}$ \\[.5mm]
            & [MW] & [\$/MWh] & [\$/h] & [\$] & [\$] & [MW/h] & [h]\\\midrule
			$g_1$ & (50\,,\,300) & 10 & 100 & 300 & 0 & 10 & 4\\\midrule
			$g_2$ & (10\,,\,150) & 20 & 100 & 70 & 0 & 40 & 3\\\midrule
			$g_3$ & (2\,,\,100) & 50 & 100 & 70 & 0 & 100 & 2\\
			\bottomrule
			\end{tabular}         
	}
    \end{center}
\end{table}
\endgroup
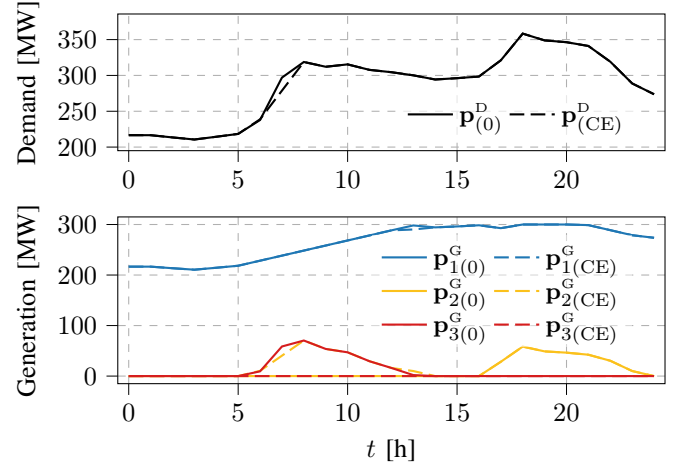
\begin{figure}[t]
    \centering
    \vspace{-4mm}
\begin{tikzpicture}

\definecolor{darkgrey176}{RGB}{176,176,176}
\definecolor{darkorange25512714}{RGB}{251, 193, 28}
\definecolor{forestgreen4416044}{RGB}{44,160,44}
\definecolor{lightgreen152223138}{RGB}{152,223,138}
\definecolor{lightgrey204}{RGB}{204,204,204}
\definecolor{lightsalmon255187120}{RGB}{255,187,120}
\definecolor{lightsteelblue174199232}{RGB}{174,199,232}
\definecolor{steelblue31119180}{RGB}{31,119,180}
\definecolor{G30}{RGB}{214,29,30}
\definecolor{G3A}{RGB}{255,152,150}

\begin{axis}[
at={(0cm,0cm)},
width=\linewidth,
height=3.4cm,
legend cell align={left},
legend style={
    fill opacity=0.7,
    draw opacity=1,
    text opacity=1,
    at={(0.95,0.1)},
    anchor=south east,
    draw=none,
    font=\small
},
legend columns=2,
tick align=outside,
tick pos=left,
x grid style={darkgrey176},
xmajorgrids,
xmin=-0.5, xmax=24.5,
xtick style={color=black},
y grid style={darkgrey176},
ylabel={Demand [MW]},
ymajorgrids,
grid style={dashed},
ymin=195, ymax=385.73858,
ytick style={color=black}
]
\addplot [thick, black]
table {%
0 216.5987
1 216.5987
2 213.4758
3 210.5637
4 214.4183
5 218.3855
6 237.7945
7 296.9609
8 318.7573
9 312.0942
10 315.409
11 307.7544
12 304.4991
13 300.1593
14 294.3144
15 296.1606
16 298.4722
17 321.1017
18 358.3493
19 348.8776
20 346.3711
21 341.0892
22 319.1494
23 288.8034
24 273.7257
};
\addlegendentry{$\bfp^{\dem}_{(0)}$}
\addplot [thick, black, dash pattern=on 5.55pt off 2.4pt]
table {%
0 216.5987
1 216.5987
2 213.4758
3 210.5637
4 214.4183
5 218.385500000669
6 238.757299986671
7 278.7573
8 318.757300001831
9 312.094200002412
10 315.409000002994
11 307.7544
12 304.4991
13 300.1593
14 294.3144
15 296.1606
16 298.4722
17 321.1017
18 358.3493
19 348.8776
20 346.3711
21 341.0892
22 319.1494
23 288.8034
24 273.7257
};
\addlegendentry{$\bfp^{\dem}_{(\rm{CE})}$}
\end{axis}

\begin{axis}[
at={(0cm,-3.1cm)},
width=\linewidth,
height=3.8cm,
legend cell align={left},
legend style={
    fill opacity=0.5,
    draw opacity=1,
    text opacity=1,
    at={(0.465,0.893)},
    anchor=north west,
    draw=none,
    font=\small
},
legend columns=2,
tick align=outside,
tick pos=left,
x grid style={darkgrey176},
xlabel={$t$ [h]},
xmajorgrids,
xmin=-0.5, xmax=24.5,
xtick style={color=black},
y grid style={darkgrey176},
ylabel={Generation [MW]},
ymajorgrids,
grid style={dashed},
ymin=-15, ymax=315,
ytick style={color=black}
]
\addplot [thick, steelblue31119180]
table {%
0 216.5987
1 216.5987
2 213.4758
3 210.5637
4 214.4184
5 218.3855
6 228.3855
7 238.3855
8 248.3855
9 258.3855
10 268.3855
11 278.3855
12 288.3855
13 298.1595
14 294.3144
15 296.1607
16 298.4722
17 292.7523
18 300
19 300
20 300
21 298.8035
22 288.8035
23 278.8035
24 273.7256
};
\addlegendentry{$\bfp^{\gen}_{1(0)}$}
\addplot [thick, steelblue31119180, dash pattern=on 5.55pt off 2.4pt]
table {%
0 216.5987
1 216.5987
2 213.4758
3 210.5637
4 214.4183
5 218.385500000669
6 228.385500000669
7 238.385500000669
8 248.385500000669
9 258.385500000669
10 268.385500000669
11 278.385500000669
12 288.385500000669
13 290.1593
14 294.3144
15 296.1606
16 298.4722
17 292.7524
18 300
19 300
20 300
21 298.8034
22 288.8034
23 278.8034
24 273.7257
};
\addlegendentry{$\bfp^{\gen}_{1(\rm{CE})}$}
\addplot [thick, darkorange25512714]
table {%
0 0
1 0
2 0
3 0
4 0
5 0
6 0
7 0
8 0
9 0
10 0
11 0
12 0
13 0
14 0
15 0
16 0
17 28.3493
18 58.3493
19 48.8775
20 46.3711
21 42.2857
22 30.3458
23 10
24 0
};
\addlegendentry{$\bfp^{\gen}_{2(0)}$}
\addplot [thick, darkorange25512714, dash pattern=on 5.55pt off 2.4pt]
table {%
0 0
1 0
2 0
3 0
4 0
5 0
6 10.3717999860022
7 40.3717999993314
8 70.3718000011625
9 53.7087000017439
10 47.0235000023254
11 29.3688999993313
12 16.1135999993313
13 10
14 0
15 0
16 0
17 28.3493
18 58.3493
19 48.8776
20 46.3711
21 42.2858
22 30.346
23 10
24 0
};
\addlegendentry{$\bfp^{\gen}_{2(\rm{CE})}$}
\addplot [thick, G30]
table {%
0 0
1 0
2 0
3 0
4 0
5 0
6 9.409
7 58.5754
8 70.3717
9 53.7088
10 47.0236
11 29.3687
12 16.1136
13 2
14 0
15 0
16 0
17 0
18 0
19 0
20 0
21 0
22 0
23 0
24 0
};
\addlegendentry{$\bfp^{\gen}_{3(0)}$}
\addplot [thick, G30, dash pattern=on 5.55pt off 2.4pt]
table {%
0 0
1 0
2 0
3 0
4 0
5 0
6 0
7 0
8 0
9 0
10 0
11 0
12 0
13 0
14 0
15 0
16 0
17 0
18 0
19 0
20 0
21 0
22 0
23 0
24 0
};
\addlegendentry{$\bfp^{\gen}_{3(\rm{CE})}$}
\end{axis}

\end{tikzpicture}
    \vspace{-7mm}
    \caption{Demand and associated optimal generation profiles for factual scenario (0) and counterfactual scenario (CE). It reveals that a slower demand ramp between 6:00\,AM and 8:00\,AM would have allowed the commitment of $g_2$ during the morning peak.}
    \label{fig:uc_example}
	\vspace{-5mm}
\end{figure}

\section{Generating CEs via Bilevel Programming}\label{sec:bilo}
In the previous examples, we assumed that we can readily access the counterfactual demand scenario, without explicitly stating how the minimal perturbation was derived. We subsequently propose approaches to generate CEs for the DCOPF and UC problems. In ML, CEs are commonly derived via gradient descent methods due to the absence of constraints \cite{wachterCounterfactualExplanationsOpening2018}. But, as presented in problem formulation \eqref{eq:ce_general}, generating CEs for optimisation problems necessitates a bilevel formulation, where the optimality of the lower-level problem is explicitly enforced. The following sections present common strategies to solve bilevel problems with linear, or mixed-integer linear problems (MILP) at the lower level.
\subsection{Linear Problems}\label{ssec:CE4LP}
	For cases where the lower-level problem is linear, solution techniques commonly resort to reformulating the problem into a single-level one, enforcing the optimality of the lower-level problem by adding optimality constraints \cite{kleinertSurveyMixedIntegerProgramming2021}. 
	These solution techniques generally start from either of the two following reformulations: (i) via the Karush-Kuhn-Tucker (KKT) conditions or (ii) via the strong duality theorem. In this paper, we use KKT conditions, while the reader is referred to \cite{kurtzCounterfactualExplanationsLinear2024} for an approach employing strong duality. Assuming both $f$ and $g$ in \eqref{seq:ce_general_opt_constr} to be continuous and convex in $x$, the KKT-based single-level reformulation of \eqref{eq:ce_general} reads
	\begin{subequations}\label{eq:ce_general_kkt}
		\begin{alignat}{3}
			\mkern-10mu(\hatparams\,,\,\mathbf{x}_{\CE}^*)\mkern5mu\in\mkern5mu&\argmin_{\params\,, \,\mathbf{x}} \mkern-20mu & & \mkern25mu\norm{\params-\tildeparams}_1\\
			& \text{s.t.} & & g_i(\mathbf{x},\params)\geq0 \mkern5mu,\mkern5mu\forall i \in \cSet\\
			& & & \nabla_{\mathbf{x}}f(\mathbf{x})-\sum_{i\in\mathcal{C}}\lambda_i\nabla_{\mathbf{x}}g_i(\mathbf{x},\params)=\bm{0}\label{seq:ce_general_statio}\\
			& & & \lambda_i\cdot g_i(\mathbf{x},\params)=0\,\,,\forall i\in\mathcal{C}\label{seq:ce_general_compl}\\
			& & & \mathbf{x} \in {\mathcal{X}}\label{seq:sol_req}\\
			& & & \params\in\mathcal{M}_{\params}\\
			& & & \mathbf{x},\bm{\lambda} \geq \bm{0}\quad,
		\end{alignat}
	\end{subequations}
    %
	where $\bm{\lambda}$ are the dual variables linked to the constraints in $\cSet$. \eqref{seq:ce_general_statio} and \eqref{seq:ce_general_compl} enforce stationarity and complementarity respectively, and thereby replace the previous optimality constraint \eqref{seq:ce_general_opt_constr}.
	A fundamental difficulty in this formulation is the complementarity constraint \eqref{seq:ce_general_compl} -- containing bilinear terms, due to the multiplication of dual variables $\bm{\lambda}$ with primal variables $\mathbf{x}$ or mutable parameters within $\params$. This turns the problem into a non-convex, non-linear optimisation problem, complicating its solution with off-the-shelf solvers. Note that this issue may be circumvented by specifying a fully fixed solution requirement, i.e., replacing \eqref{seq:sol_req} with $\mathbf{x}=\mathbf{x}^*_{\CE}$. Indeed, the resulting inverse optimisation problem would then also be linear \cite{ahujaInverseOptimization2001}. This would, however, come with a profound loss of generality -- even more for problems with many decision variables. The following paragraphs present commonly used strategies to address the non-linear nature of the above problem.
	\paragraph{Special Ordered Sets} The complementarity conditions within \eqref{seq:ce_general_compl} tell us that for any constraint $i$, both $\lambda_i$ and $g_i(\mathbf{x}, \params)$ cannot simultaneously take non-zero values. This property can be enforced by leveraging so-called \textit{Special Order Sets} of type 1, a special class of constraints that allows only one member of a set of variables to assume a non-zero value. Constraints of this type can directly be defined in off-the-shelf solvers such as \cite{gurobioptimizationGurobiOptimizerReference2024}. This method is later referred to as \texttt{SOS1}.

 	\paragraph{Fortuny-Amat-McCarl Linearisation} Another standard approach is based on the {Fortuny-Amat-McCarl linearisation}, which relies on introducing binary variables, thereby transforming the problem into a Mixed Integer Problem (MIP). 
	It builds on the notion that complementarity constraints can be replaced by ``\textit{big-$M$}'' constraints. The complementarity condition \eqref{seq:ce_general_compl} is hereby replaced by two auxiliary constraints
	\begin{subequations}
		\begin{equation}\label{seq:fortuny1}
		0\leq\lambda_i \leq z_i M_i\,\,,\,\,\forall i\in\cSet
		\end{equation}
		\begin{equation}\label{seq:fortuny2}
		0\leq g_i(\mathbf{x},\params)\leq (1-z_i) M_i\,\,,\,\,\forall i\in\cSet\quad,
		\end{equation}
	\end{subequations}
	where $z_i$ is a binary variable introduced for each inequality constraint. A drawback is that the choice of appropriate values for $M$ is a challenging task in itself \cite{pinedaSolvingLinearBilevel2019}. Chosen too small and the global optimum could be excluded from the feasible region. Conversely, when set too large, numerical instabilities may occur. This paper follows the procedure explained in \cite{pratLearningActiveConstraints2023}. For the primal variables, the $M$ values are directly obtained from the given constraint bounds. Since the dual variables are not explicitly bounded, we assume that we can access a database of previous runs from where we can find the maximum $\lambda_i$ ever recorded. These maximum values are then multiplied with a safety factor of 10 to obtain the respective $M$ values. We henceforth refer to this method as \texttt{MIP}.

    \subsection{Mixed-Integer Linear Problems}\label{ssec:CE4MILP}
	We now study the case where a vector of binary variables $\bf{u}$ is included in both objective function $f$ and constraints $g_i$. In contrast to the bilevel linear problem, the reformulation via KKT conditions or strong duality theorem is rendered invalid in the presence of integer variables. Due to the lack of easily implementable lower-level optimality constraints for MILPs, the range of exact solution algorithms is much sparser than for the linear case. Existing methods to globally solve bilevel MILP problems usually build on the following lower-level inequality formulation (here adapted to immediately fit the CE problem structure):
	\begin{subequations}\label{eq:ce_milp}
		\begin{alignat}{2}
			\min_{\params\,, \,\bf{x}\,,\bf{u}} \mkern5mu & \norm{\params-\tildeparams}_1 &\\
			\text{s.t.} \mkern5mu & \mkern5mu g_i(\bf{x},\bf{u}, \params)\geq0 \mkern5mu,\mkern5mu\forall i \in \cSet\\
			& f(\bf{x},\bf{u}, \params)\leq\min_{\bf{y},\bf{v}} \left.\begin{cases}
				f(\bf{y},\bf{v},\params): \\
				\bf{g}(\bf{y},\bf{v},\params)\geq \bm{0}
			\end{cases} \mkern-20mu\right\}\mkern2mu,\forall\bf{v}\in\intSet(\params)\label{seq:ce_milp_opt_ineq}\\
			& (\bf{x},\bf{u}) \in {\mathcal{X}}\label{seq:ce_milp_sol_req}\\
			& \params\in\mathcal{M}_{\params}\mkern5mu.
		\end{alignat}
	\end{subequations}
	Here, $\intSet(\params)$ is the set of \textit{all possible} combinations of binary decision vectors yielding a feasible solution for a given set of parameters $\params$. Hence, if constraint \eqref{seq:ce_milp_opt_ineq} is satisfied for all vectors $\bf{v}\in\intSet(\params)$, and the respective continuous variables $\bf{y}$ are forced to assume minimising values (e.g., via KKT conditions), then the resulting triplet $(\hatparams\,,\,\bf{x}_{\CE}^*,\bf{u}_{\CE}^*)$ is a global minimiser to \eqref{eq:ce_milp}. 
	
	Clearly, the upper formulation may not be implementable in practice, due to the exponentially large number of combinations within $\intSet$. To address this inherent computational difficulty, we implement the decomposition algorithm proposed in \cite{yueProjectionbasedReformulationDecomposition2019}, which iteratively adds new unique combinations of $\bf{u}$ to $\intSet$. The decomposition algorithm splits the bilevel MILP into a master problem and a subproblem. Within each iteration, the master problem solves \eqref{eq:ce_milp} for the currently minimal parameter variation of $\params$, which are then fixed in the subproblem. The subproblem is then used find new combinations of $\bf{u}$. By doing so, the algorithm successively grows the lower bound of $\norm{\params-\tildeparams}_1$ until a parameter variation $\params$ is found that produces a MILP solution $(\bf{x},\bf{u})$ which inherently satisfies \eqref{seq:ce_milp_sol_req} (at which point the constraint could as well be dropped). Due to space constraints, we do not contain a full outline of the algorithm, for which we refer the reader to the authors' original work. This method is hereafter called \texttt{DECOMP}.

	\subsection{Data-Driven Heuristic Cuts}
	Operators usually solve OPF or UC problems repeatedly, often as frequent as every 15 minutes \cite{zamzamLearningOptimalSolutions2019}. Adding to a wealth of studies ($\mkern-5mu$\cite{zamzamLearningOptimalSolutions2019,chenLearningSolveDCOPF2022,pinedaLearningUnitCommitment2022,xavierLearningSolveLargeScale2019}, to name a few) that have demonstrated performance improvements by using data for both DCOPF and UC problems, this paper aims to leverage historical data to speed up the solution of \eqref{eq:ce_general_kkt} and \eqref{eq:ce_milp}. 
	
	We consider the linear case first. Assume that the lower-level LP has been solved for $S$ times, each time varying the input parameters. A dataset $\dataSet=\{\params,\mathbf{x}^*,\bm{\ac}\}_{j=1}^{S}$ is thereby obtained, where $\bm{\ac}\in\{0,1\}^{|\cSet|}$ indicates whether a constraint is binding or not. While in theory, the number of possible combinations of $\bm{\ac}$ could be exponentially large, it has been shown that for many DCOPF problems, only a few combinations actually occur in practical settings \cite{misraLearningConstrainedOptimization2019}. This also manifests itself in the existence of subsets $\cSet_1$, $\cSet_0\subseteq\cSet$ at whose indices the constraints are always active or inactive, i.e., $\alpha_i^{[j]}=1\,,\,\forall i\in\cSet_1\,,\,\forall j\in[S]$ and $\alpha_i^{[j]}=0\,,\,\forall i\in\cSet_0\,,\,\forall j\in[S]$. We use this observation by adding cuts to the Fortuny-Amat-McCarl linearisation -- fixing $z_i=1$ if $i\in\cSet_1$ or $z_i=0$ if $i\in\cSet_0$. The goal of this method is to reduce the number of binary variable combinations and thereby improve the solution speed. This method is subsequently called \texttt{MIP+cut}.

	For a lower-level MILP the access to dual variables is prohibited due to non-convexity.
	But similarly as for the active constraint indicators, one may also find subsets of the binary variables that remain constant throughout all recorded samples, i.e., $u_i^{[j]}=1\,,\,\forall j\in[S]$ and $u_i^{[j]}=0\,,\,\forall j\in[S]$. In the master problem of the decomposition algorithm, these variables are fixed to either 0 or 1, resulting in cuts on the binary decision space. 
	An additional step to improve the performance is taken by bounding the objective function term $\norm{\params_{\sz}-\params_{\CE}}_1$. This is done by finding the single sample within $\dataSet$ for which $\params$ is minimally distanced from the factual instance, and add a constraint to bound CE objective function from above. The method resulting from these data-driven cuts and bounds for the decomposition algorithm is referred to as \texttt{DECOMP+cut}.

	Unlike other data-assisted optimisation methodologies, the above two methods avoid the sometimes non-trivial process of choosing the right ML model (and training it with a suitable set of hyperparameters). Instead, this approach purely relies on the information contained in the dataset. 
	%
	Problems arise when the solution region $\mathcal{X}$ is specified in a way that the optimal CE solution $\{\params_{\CE}, \mathbf{x}_{\CE}^*\}$ to \eqref{eq:ce_general_kkt} (or $\{\params_{\CE}, \mathbf{x}_{\CE}^*,\bf{u}_{\CE}^*\}$ to \eqref{eq:ce_milp}) would require to produce a new, unseen $\widetilde{\bm{\ac}}\notin\dataSet$ that contradicts $\cSet_1$ or $\cSet_0$, which would inhibit the recovery of the nearest CE parameters. In the worst case, one cannot find any feasible solution, despite its existence. Nevertheless, this problem would prevail for any ML surrogate aimed at approximating the optimal solution policy.

\subsection{Upper bound for CE Objective}

Both for DCOPF and UC CE problems we rely on MILP formulations to implement the aforementioned solution methods. State-of-the-art MILP solvers generally employ branch-and-bound techniques, which can substantially benefit when provided additional objective function bounds. We thus make use of the following valid upper bound.

\begin{proposition}\label{proposition}
	For both DCOPF and UC CE problems, whose objective is to minimise the distance $\norm{\bfp^{\dem}-\bfp^{\dem}_{\sz}}_1$, subject to fulfilling a solution requirement in line with assumption \ref{as:X_DCOPF_UC}, the following inequality holds at optimality
	\begin{equation}
		\norm{\bfp^{\dem}_{\CE}-\bfp^{\dem}_{\sz}}_1\leq\norm{\bfp^{\gen *}_{\CE}-\bfp^{\gen *}_{\sz}}_1.
	\end{equation}
\end{proposition}

\begin{proof}\renewcommand{\qedsymbol}{}
Provided in the appendix.
\end{proof}

In other words, the demand variation needed to satisfy a user-defined requirement on the optimal dispatch is at most equal to the resulting variation in the dispatch solution space. This inequality was thus added as an explicit constraint to all of the methods introduced so far.

\section{Numerical Experiments}\label{sec:res}

In a real-world application, the developed explainability framework would be used on an ad-hoc basis, whenever an unanticipated market outcome is suggested by the solver. But to rigorously test the developed methods, we need to analyse the framework's performance across different network topologies, question formulations, and demand scenarios.

\subsection{Experimental Setup}

\subsubsection{Network data}
\textbf{DCOPF}: The network models used to study explanations for DCOPF problems are based on the benchmark systems in \cite{BenchmarksOptimalPower2021}. Specifically, the \textit{Active Power Increase} variants were used (files appended by the \texttt{\_api} suffix), as these are specifically altered to increase the probability of constrained scenarios, and conditions that one would expect to produce unintuitive optimal solutions. 
Generators which could only provide reactive power were removed due to their irrelevance for DCOPF simulations. The 9 considered test cases are summarised in Table \ref{tab:test_cases}. For simplicity, transformers are here included in the set of transmission lines $\eSet$.

\textbf{Unit Commitment}: Not all the test cases in \cite{BenchmarksOptimalPower2021} include UC-specific data. Therefore, complementary generator data from \cite{santosxavierUnitCommitmentjlJuliaJuMP2024} was used. For each generator $g$ in the dataset, the respective production costs ($c_{g1}^{\rm{PR}}$ and $c_{g0}^{\rm{PR}}$ were derived from the production cost curve), start-up and shut-down costs ($c_g^{\rm{SU}}$, $c_g^{\rm{SD}}$), generator limits ($\uli{p}_g$, $\oli{p}_g$), ramping limits ($R_{g}^{\uarrow}$, $R_{g}^{\darrow}$) and minimum up- and downtimes ($T^{\rm{U}}_g$, $T^{\rm{D}}_g$) were adopted. UC problems were solved for the 6 test cases marked in Table \ref{tab:test_cases}, while no UC data was available for the other networks. The numbers of constraints shown in brackets in Table \ref{tab:test_cases} are those obtained when considering UC problems with time horizons of 24 hours.
\begingroup
\setlength{\tabcolsep}{4pt}
{
\renewcommand{\arraystretch}{1.3}
\begin{table}[]
        \caption{Summary of the network characteristics of the test cases used. Data for the UC problem that differed from the DCOPF data 
        are shown in brackets.}
        \label{tab:test_cases}
    \begin{center}
    \vspace{-4mm}
    \resizebox{\linewidth}{!}{
        \begin{tabular}{lHlllHHcc}
            \toprule
            \textbf{Test case} & \textbf{\# buses} & \textbf{\# lines} & \textbf{\# units} & \textbf{\# constraints} & \# regions $|\mathcal{R}|$ & solution time [s] & \textbf{DCOPF} & \textbf{UC} \\
            \midrule\\[-5mm]
            \texttt{case5\_pjm} & 5 & 6 & 5 & \vphantom{28}39 & \vphantom{16}19 & 5.963 & \checkmark & \\
            \cline{1-9}
            \texttt{case14\_ieee} & 14 & 20 & 2 (5) & \vphantom{59}98 (777) & 5 & 6.576 & \checkmark & \checkmark \\
            \cline{1-9}
            \texttt{case24\_ieee\_rts} & 38 & 20 & 32 & \vphantom{165}240 & $>$ \textit{3,874} & timeout (1 hour) & \checkmark & \\
            \cline{1-9}
            \texttt{case30\_ieee} & 30 & 41 & 2 (6) & \vphantom{117}198 (966) & 15 & 11.99  & \checkmark & \checkmark\\
            \cline{1-9}
            \texttt{case39\_epri} & 39 & 46 & 10 & \vphantom{152}243 & $>$ \textit{2,268} & timeout (1 hour) & \checkmark & \\
            \cline{1-9}
            \texttt{case57\_ieee} & 57 & 80 & 4 (7) & 385 (1,155) & $>$ \textit{2,268} & timeout (1 hour) & \checkmark & \checkmark\\
            \cline{1-9}
            \texttt{case89\_pegase} & 89 & 210 & 11 (12) & 951 (2,052) & $>$ \textit{2,268} & timeout (1 hour)  & \checkmark & \checkmark\\
            \cline{1-9}
            \texttt{case118\_ieee} & 118 & 186 & 19 (54) & 900 (9,318) & $>$ \textit{2,268} & timeout (1 hour)  & \checkmark & \checkmark\\
            \cline{1-9}
            \texttt{case300\_ieee} & 300 & 411 & 57 (69) & 2,058 (12,057)\hspace*{-3mm} & $>$ \textit{2,268} & timeout (1 hour)  & \checkmark & \checkmark\\[-.5mm]
            \bottomrule
            \end{tabular}
        }
        \end{center}
\end{table}
}
\endgroup

\subsubsection{Demand variations}
\textbf{DCOPF}: For each test case, $S=5,000$ samples of randomly distributed demand vectors were generated. 
The minimum and maximum bounds for the demand were chosen as 0\% and 120\% of the default values $\widetilde{\bfp}^{\dem}$ provided in \cite{BenchmarksOptimalPower2021}, i.e., $\bfp^{\dem}\sim\mathrm{unif}(\bm{0},1.2\times\widetilde{\bfp}^{\dem})$. Infeasible demand scenarios were discarded, and new demand variations were generated until $5,000$ scenarios with feasible solutions were obtained.

\textbf{Unit Commitment}: For the UC, the number of samples was limited to $S=500$, where each sample contains the 24-hour demand profiles ($\bf{P}^{\dem}$) and the resulting optimal generation schedule ($\bf{P}^{\gen *},\bf{U}^{*}$). This simulates the size of a dataset obtained from a record of around one and a half years, which is a sufficiently long period to capture seasonal effects. At the same time, it is short enough to not expect major network topology changes which would require an adjustment of the model. Each 24-hour profile is produced by following the procedure detailed in \cite{xavierLearningSolveLargeScale2019}, which, for the sake of brevity, is omitted here. The only modification is that the daily peak demand is obtained by drawing from $\rm{unif}(0.6, 1.4)\times0.6\times\sum_{g\in\gSet}\oli{p}_g$, where the limits were chosen more drastically than in \cite{xavierLearningSolveLargeScale2019} to allow for a more diverse set of scenarios.
An example of the 500 profiles of a single demand resulting from this sampling procedure is shown in Figure \ref{fig:load_profiles}. The colour represents the empirical probability density of each hourly demand taking on a value in the respective power interval. 
\begin{figure}[t]
    \centering
    \vspace{0mm}
    \includegraphics[width=\linewidth]{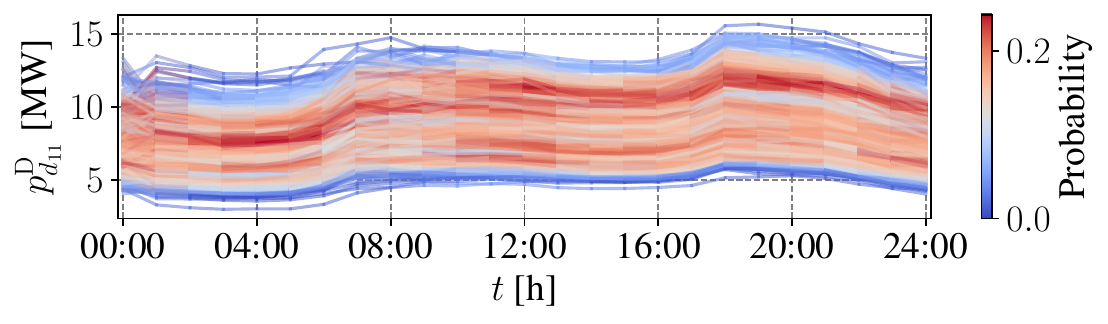}
    \vspace{-6mm}
    \caption{Hourly demand profile samples with empirical probability density for an exemplary bus in \texttt{case14}}
    \label{fig:load_profiles}
    \vspace{-3mm}
\end{figure}


\subsubsection{Question generation}\label{ssec:questions}
To effectively test the performance of the developed methods, a set of $Q=200$ questions for each test case is generated.
The process is aimed to resemble the real scenario in which a ``skip'' occurs, i.e., an out-of-merit dispatch decision (described in Section \ref{sec:intro}). The questions for DCOPF and UC are based on the definitions of the solution region $\mathcal{X}$ described in Assumption \ref{as:sol_reg_feas}. 

\textbf{DCOPF}: Based on \eqref{eq:X_DCOPF} we need to formulate a suitable minimum generator threshold $P_{g'}$ for each question. Among the 5,000 scenarios for each test case, we aim to find those that produced the most counterintuitive solutions. These are obtained by iterating over each dataset entry, where for a given sample $i\in\{1,\ldots,S\}$, we measure each unit's downward deviation from its mean dispatch value across the entire dataset, i.e., $p^{\gen}_{g,\rm{mean}} - p^{\gen[i]}_g$. If~there are no positive downward deviations (i.e., all generators are dispatched above their mean value), continue to $i+1$. Otherwise, choose $g'=\argmax_{g\in\gSet}\left\{p^{\gen}_{g,\rm{mean}} - p^{\gen[i]}_g\right\}$, i.e., the unit with the largest downward deviation from their mean value. The dispatch threshold is then defined as a random value between the actual dispatch level and the maximum value that was recorded for $g'$ within the entire dataset $\dataSet$, i.e., $P_{g'}\sim\rm{unif}\left(p^{\gen[i]}_{g'}, \max_{j\in[S]}\left\{p^{\gen[j]}_{g'}\right\}\right)$. Finally, since $S>Q$, pick $Q$ random samples from the candidate questions.

\textbf{Unit Commitment}: As described in \eqref{eq:X_UC} the solution region $\mathcal{X}$ is characterised only by the constraint $u_{g'\mkern-3mu,t'}=1$. 
The procedure for generating $Q=200$ questions is done in analogy to the DCOPF. Iterating over each sample $i\in\{1,\ldots,S\}$, the largest downward deviation of each generator's hourly commitment variable from the respective mean $u_{g,t,\rm{mean}}=1/{S}\sum_{j\in[S]} u_{g,t}^{[j]}$ is found. Then, the tuple $\{g',t'\}=\argmax_{g\in\gSet,t\in\tSet}\left\{u_{g,t,\rm{mean}} - u_{g,t}\right\}$ is picked, and  $u_{g'\mkern-3mu,t'}=1$ is added as a question candidate. Again, $Q$ random samples are drawn from the set of question candidates.


\subsection{Performance Metrics}

\paragraph{Speed}
The ability to quickly generate explanations would be highly desirable in real-life applications. If, for instance, a control room engineer is requested to justify their dispatch decision to affected stakeholders, a long delay could possibly add to the already existing time pressure in such environments. Furthermore, different generators may simultaneously ask for explanations, which could lead to a prohibitively long queue of CE generation problems. 

Let $T_{\rm{run}}$ denote the runtime. For CE problems addressing DCOPF instances, $T_{\rm{run}}$ is measured for all of the proposed methods by recording the wallclock time of the solver. In the case of the UC CE problems, $T_{\rm{run}}$ represents the total time required by the decomposition algorithm, and hence also includes the time of adding new cuts to the master problem. An upper time limit of 10 minutes is set on the solution time for both DCOPF and UC CE problems.

\paragraph{Minimality}
Based on the definition in \ref{sec:XOpt}, CEs are \textit{minimal} changes to the factual parameters such that a user-defined solution region yields optimal solutions.
We define two indicators to measure the minimality of a CE. For a given counterfactual demand scenario $\bfp^{\dem}_{\CE}$ the first indicator, termed \textit{peak-normalised distance} (PND), is computed as
    \begin{equation}\label{eq:delta_pnd}
        \Delta_{\rm{PND}}^{\dem}\left(\bfp^{\dem}_{\CE}\right) = \frac{\norm{\bfp^{\dem}_{\sz}-\bfp^{\dem}_{\CE}}_1}{P_{\max}^{\dem}}\times100\%
    \end{equation}
where $P_{\max}^{\dem}$ is the single maximum net system demand that occurred among all of the $S$ samples.
%
To further obtain a meaningful estimate of how each method performs w.r.t. a primitive benchmark method, let
\begin{equation}
    \Delta_{\knn}^{\dem}\left(\bfp^{\dem}_{\CE}\right)=\frac{\norm{\bfp^{\dem}_{\sz}-\bfp^{\dem}_{\CE}}_1}{\norm{\bfp^{\dem}_{\sz}-\bfp^{\dem}_{\knn}}_1}\times100\%
\end{equation}
where $\bfp^{\dem}_{\knn}=\min_{i\in[S]}\bigl\{\bigl\lVert\bfp^{\dem}_{\sz}-\bfp^{\dem[i]}\bigr\rVert_1\,\,:\,\,p_{g'}^{\gen[i]}>P_{g'}\bigr\}$ is the counterfactual demand variation which one would obtain when simply choosing the demand sample nearest to $\bfp^{\dem}_{\sz}$ that satisfies the solution requirement on $p^{\gen}_{g'}$ for the DCOPF and on $u_{g'}$ for the UC, i.e., carrying out a $k$-nearest neighbours look-up with $k=1$. After all, if the solutions to the algorithmic approaches lead to $\Delta_{\knn}^{\dem}$ values mostly close to 100\%, then their usage would be hardly justifiable, since the $k$NN1 method is much easier to implement. Evidently, the lower both $\Delta_{\rm{PND}}^{\dem}$ and $\Delta_{\knn}^{\dem}$ of a CE, the better the tested method performed on answering the respective question.

\subsection{Results}

Based on the above setup, the methods proposed in Section \ref{sec:bilo} were tested on their ability to quickly generate minimal CEs for dispatch decisions in the DCOPF and UC problems. 
All experiments were run on an Intel Core i7 CPU with a clock rate of 2.6 GHz and 32GB of RAM. DCOPF and UC optimisation problems were modelled using the \textit{Optimisation and Analysis Toolbox for Power Systems} \cite{bukhshOATSOptimisationAnalysis2020}, and then solved using {Gurobi} v11.0.3 \cite{gurobioptimizationGurobiOptimizerReference2024}.
%
\begin{figure}[t]
    \centering
    \vspace{0mm}
    \includegraphics[width=\linewidth,angle=0,trim=2mm 0mm 5mm 2mm]{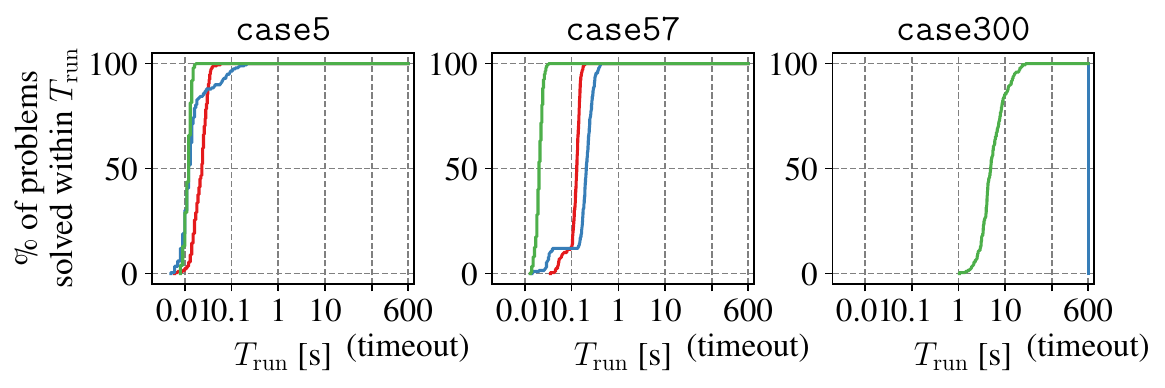}\\[2mm]
    \includegraphics[width=\linewidth,angle=0,trim=0mm 0mm 0mm 2mm]{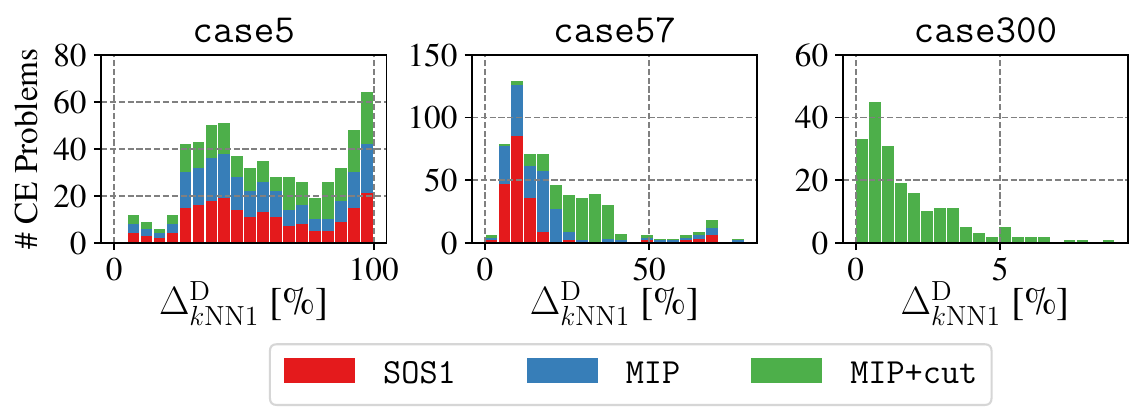}
    \vspace{-4mm}
    \caption{Cumulative proportions of runtimes and empirical distributions of peak-normalised distances ($\Delta_{\knn}^{\dem}$) of DCOPF CE problems on selected test cases}
    \label{fig:dcopf_PND}
    \vspace{-3mm}
\end{figure}
\subsubsection{Speed}
\textbf{DCOPF}: Figure \ref{fig:dcopf_PND} and Table \ref{tab:dcopf_performance} compare the runtimes of the three tested methods \texttt{SOS1}, \texttt{MIP} and \texttt{MIP+cut} on different DCOPF test cases. Note that, while all questions were generated in the same systematic fashion as described in Section \ref{ssec:questions}, the difficulty of solving a CE problem depends on both the factual scenario and the solution region defined by the user. The heterogeneity of the questions is what leads to different distribution shapes observable in the upper plots in Figure \ref{fig:dcopf_PND}, where runtimes for problems using the same method and test case range from a few seconds to several minutes.

Unsurprisingly, there is a trend of CE problems applied to larger networks leading to longer runtimes. While all methods maintain fast solution speed throughout the smallest 6 test cases, it is also found that, for larger networks, the exact solution via SOS1-constraints starts taking longer than desirable. This observation aligns well with findings from other studies exploring bilevel formulations with lower-level DCOPF problems (e.g., \cite{pinedaSolvingLinearBilevel2019}). For instance, for \texttt{case300}, neither \texttt{SOS1} nor \texttt{MIP} are able to converge within 10 minutes. As hoped, the cuts introduced by \texttt{MIP+cut} can notably boost the computational efficiency in such cases. Even at maximum runtime, all solutions can be obtained within a fraction of the 10-minute time limit.
\begin{figure}[t]
    \centering
    \vspace{0mm}
    \includegraphics[width=\linewidth,angle=0,trim=2mm 0mm 5mm 2mm]{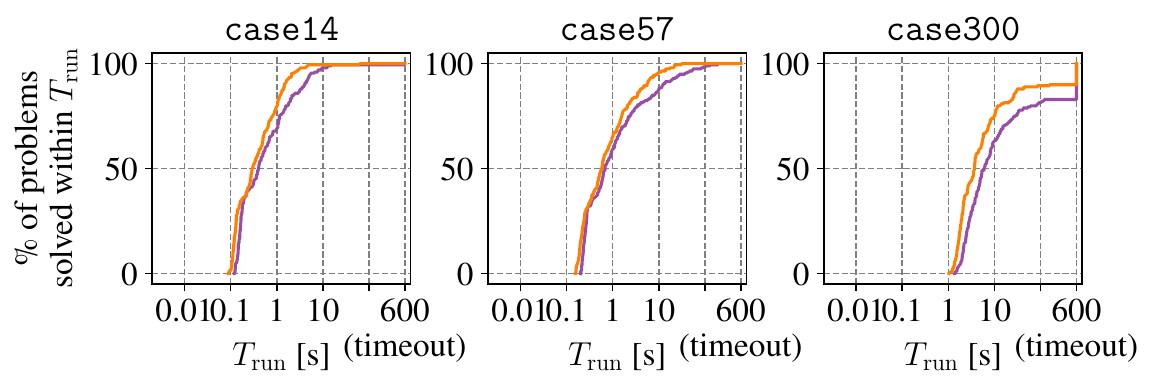}\\[2mm]
    \includegraphics[width=0.98\linewidth,angle=0,trim=0mm 0mm 0mm 3mm]{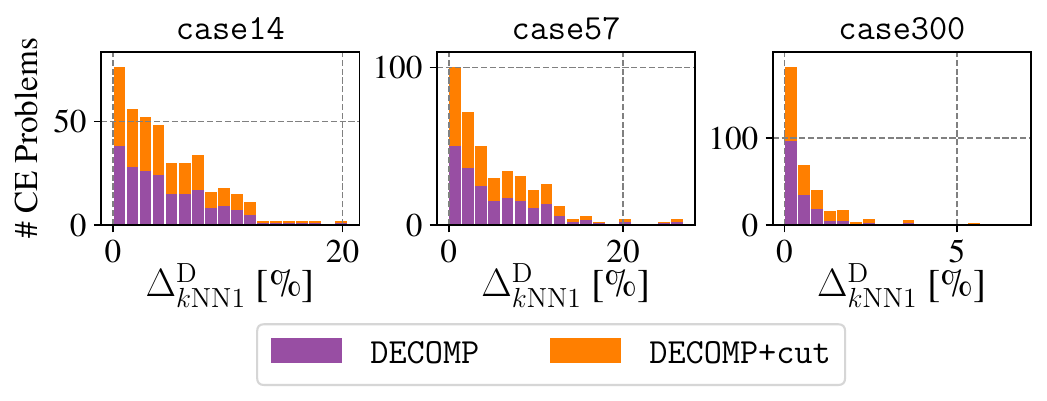}
    \vspace{-0.5mm}
    \caption{Cumulative proportions of runtimes and empirical distributions of peak-normalised distances ($\Delta_{\knn}^{\dem}$) of UC CE problems on selected test cases}
    \label{fig:uc_PND}
    \vspace{-5mm}
\end{figure}

\textbf{Unit Commitment}: Runtimes for the 200 CE generation problems for different UC problems are presented in Figure \ref{fig:uc_PND} and Table \ref{tab:uc_performance}. As in the DCOPF, we observe that larger UC instances generally result in computationally more demanding CE problems. Nevertheless, there are some compelling differences compared to the DCOPF CE problems. 

Firstly, despite the difficulties imposed by the presence of binary variables in the UC problem, it can be seen that even for the 300-bus system, more than 80\% of all problems still diverge to an optimal CE solution within the 10-minute time limit for both exact (\texttt{DECOMP}) and data-assisted (\texttt{DECOMP+cut}) methods. This is different from the DCOPF problem where, for the same 300-bus network, the exact \texttt{SOS1} method cannot produce any optimal CE solutions within the time limit. A plausible explanation lies in the way the UC problem is formulated. Since zonal transmission constraints are not considered, there is no need to introduce mutable demand parameters at every node. As such, the number of variables appearing in the objective function of the CE problem is agnostic to the size of the network. Specifically, the distance term to be minimised only contains the up- and downwards deviations of the regionally aggregated demand over the given time period $\tSet$ (24 hours). Certainly, a more detailed UC formulation requiring the minimisation of all \textit{nodal} demand deviations may thus quickly render the decomposition algorithm intractable.

A second insight is that the speed improvements added by the data-informed cuts in \texttt{DECOMP+cut} are not as substantial as it was the case for \texttt{MIP+cut}. For instance, the 1.25x performance gain achieved via the data-driven method for \texttt{case89} is rather marginal. Further studies need to be carried to better interpret this behaviour, in particular with respect to the mathematical structure of the decomposition algorithm itself. A key factor could be the iterative nature of the decomposition algorithm, alternating between solving the master problem and the UC subproblem. While the cuts allow to solve both problems more efficiently, there is no reason to expect that the cuts will consistently reduce the number of iterations needed.
Nevertheless, \texttt{DECOMP+cut} still manages to converge to a CE on average 3.77\% more often than \texttt{DECOMP} and may thus be used as a fall-back option in case \texttt{DECOMP} does not converge within the given time limit for larger test cases.

%

\begingroup
\setlength{\tabcolsep}{4pt}
\begin{table}[t]
   \caption{Statistics of solving 200 DCOPF CE problems on various test cases with different methods (best methods in \textbf{boldface})}
    \label{tab:dcopf_performance}
    \begin{center}
   \vspace{-2mm}
    \resizebox{0.99\linewidth}{!}{
        \begin{tabular}{m{1.1cm}lHHH|l|HHHHll|Hll|HllH}
            \toprule
            \multirow[c]{2}{*}{\textbf{Test case}} & \multirow[c]{2}{*}{\textbf{Method}} &  \# CEs & \multirow[c]{2}{*}{\textbf{\# inf}} & \multirow[c]{2}{*}{\textbf{\# Timeout}} & \multirow[c]{2}{*}{\textbf{\# Succ.}} & \multirow[c]{2}{*}{\textbf{\% inf}} & \multirow[c]{2}{*}{\textbf{\shortstack[l]{\% time-\\out}}} & \% Success & & \multicolumn{2}{c|}{$T_{\rm{run}}$ [s]} & & \multicolumn{2}{c|}{$\Delta_{\rm{PND}}^{\dem}$ [\%]} & & \multicolumn{2}{c}{$\Delta_{\knn}^{\dem}$ [\%]} & \\[0.5mm]
            & & & & & & & & & min & \multicolumn{1}{c}{mean} & \multicolumn{1}{c|}{max} & {min} & \multicolumn{1}{c}{mean} & \multicolumn{1}{c|}{max} & min & \multicolumn{1}{c}{mean} & \multicolumn{1}{c}{max} & mean ADI [MW] \\\midrule
            \multirow[c]{3}{*}{\texttt{case5}} & \texttt{SOS1} & 200 & 0 & 0 & 200 & 0.00 & 0.00 & 100.00 & 0.006 & 0.023 & 0.107 & \textbf{0.125} & \textbf{9.618} & \textbf{38.327} & \textbf{6.016} & \textbf{57.027} & \textbf{99.873} & 5.955 \\
            & \texttt{MIP} & 200 & 0 & 0 & 200 & 0.00 & 0.00 & 100.00 & \textbf{0.005} & 0.023 & 0.252 & \textbf{0.125} & \textbf{9.618} & \textbf{38.327} & \textbf{6.016} & \textbf{57.027} & \textbf{99.873} & 5.955 \\
            & \texttt{MIP+cut} & 200 & 0 & 0 & 200 & 0.00 & 0.00 & 100.00 & 0.008 & \textbf{0.012} & \textbf{0.017} & \textbf{0.125} & 11.211 & \textbf{38.327} & \textbf{6.016} & 62.724 & \textbf{99.873} & \textbf{4.362} \\\midrule
            \multirow[c]{3}{*}{\texttt{case14}} & \texttt{SOS1} & 200 & 0 & 0 & 200 & 0.00 & 0.00 & 100.00 & 0.005 & 0.02 & 0.042 & \textbf{0.1} & \textbf{18.963} & \textbf{56.725} & \textbf{0.901} & \textbf{72.651} & \textbf{99.891} & 6.422 \\
            & \texttt{MIP} & 200 & 0 & 0 & 200 & 0.00 & 0.00 & 100.00 & 0.007 & 0.011 & 0.017 & \textbf{0.1} & \textbf{18.963} & \textbf{56.725} & \textbf{0.901} & \textbf{72.651} & \textbf{99.891} & 6.422 \\
            & \texttt{MIP+cut} & 200 & 0 & 0 & 200 & 0.00 & 0.00 & 100.00 & \textbf{0.003} & \textbf{0.007} & \textbf{0.012} & \textbf{0.1} & 19.671 & \textbf{56.725} & \textbf{0.901} & 74.57 & \textbf{99.891} & \textbf{5.713} \\\midrule
           \multirow[c]{3}{*}{\texttt{case24}} & \texttt{SOS1} & 200 & 0 & 0 & 200 & 0.00 & 0.00 & 100.00 & \textbf{0.104} & \textbf{1.65} & \textbf{10.056} & \textbf{0.0} & \textbf{2.08} & \textbf{11.361} & \textbf{0.0} & \textbf{8.815} & \textbf{31.962} & 19.424 \\
            & \texttt{MIP} & 200 & 0 & 0 & 200 & 0.00 & 0.00 & 100.00 & 0.153 & 3.763 & 20.262 & \textbf{0.0} & 2.217 & 13.542 & \textbf{0.0} & 9.264 & 38.439 & \textbf{19.288} \\
            & \texttt{MIP+cut} & 200 & 0 & 0 & 200 & 0.00 & 0.00 & 100.00 & 0.172 & 3.771 & 15.881 & \textbf{0.0} & 2.217 & 13.542 & \textbf{0.0} & 9.264 & 38.439 & \textbf{19.288} \\\midrule
           \multirow[c]{3}{*}{\texttt{case30}} & \texttt{SOS1} & 200 & 0 & 0 & 200 & 0.00 & 0.00 & 100.00 & 0.009 & 0.056 & 0.085 & \textbf{0.291} & \textbf{5.805} & \textbf{12.219} & \textbf{1.978} & \textbf{18.992} & \textbf{60.926} & 28.287 \\
            & \texttt{MIP} & 200 & 0 & 0 & 200 & 0.00 & 0.00 & 100.00 & 0.017 & 0.057 & 0.102 & \textbf{0.291} & 18.123 & 33.175 & \textbf{1.978} & 51.098 & 78.098 & 15.969 \\
            & \texttt{MIP+cut} & 200 & 0 & 0 & 200 & 0.00 & 0.00 & 100.00 & \textbf{0.006} & \textbf{0.009} & \textbf{0.017} & \textbf{0.291} & 21.098 & 55.383 & \textbf{1.978} & 56.5 & 81.433 & \textbf{12.994} \\\midrule
           \multirow[c]{3}{*}{\texttt{case39}} & \texttt{SOS1} & 200 & 0 & 0 & 200 & 0.00 & 0.00 & 100.00 & 0.088 & 0.647 & 1.727 & \textbf{0.071} & \textbf{5.398} & \textbf{16.247} & \textbf{0.541} & \textbf{21.058} & \textbf{51.983} & 18.308 \\
            & \texttt{MIP} & 200 & 0 & 0 & 200 & 0.00 & 0.00 & 100.00 & 0.167 & 1.309 & 5.376 & \textbf{0.071} & 5.485 & \textbf{16.247} & \textbf{0.541} & 21.307 & \textbf{51.983} & 18.221 \\
            & \texttt{MIP+cut} & 200 & 0 & 0 & 200 & 0.00 & 0.00 & 100.00 & \textbf{0.039} & \textbf{0.246} & \textbf{1.084} & \textbf{0.071} & 5.881 & 18.141 & \textbf{0.541} & 22.58 & \textbf{51.983} & \textbf{17.825} \\\midrule
            \multirow[c]{3}{*}{\texttt{case57}} & \texttt{SOS1} & 200 & 0 & 0 & 200 & 0.00 & 0.00 & 100.00 & 0.035 & 0.124 & 0.24 & \textbf{0.204} & \textbf{4.313} & \textbf{22.385} & \textbf{1.097} & \textbf{14.752} & \textbf{78.902} & 28.308 \\
            & \texttt{MIP} & 200 & 0 & 0 & 200 & 0.00 & 0.00 & 100.00 & 0.014 & 0.205 & 0.433 & \textbf{0.204} & 5.731 & \textbf{22.385} & \textbf{1.097} & 18.402 & \textbf{78.902} & 26.89 \\
            & \texttt{MIP+cut} & 200 & 0 & 0 & 200 & 0.00 & 0.00 & 100.00 & \textbf{0.013} & \textbf{0.021} & \textbf{0.033} & \textbf{0.204} & 10.219 & 23.124 & \textbf{1.097} & 30.519 & \textbf{78.902} & \textbf{22.402} \\\midrule
            \multirow[c]{3}{*}{\texttt{case89}} & \texttt{SOS1} & 200 & 0 & 0 & 200 & 0.00 & 0.00 & 100.00 & 1.897 & 5.908 & 14.578 & \textbf{0.063} & \textbf{2.652} & \textbf{8.307} & \textbf{0.247} & \textbf{10.14} & \textbf{31.117} & 23.588 \\
            & \texttt{MIP} & 200 & 0 & 0 & 200 & 0.00 & 0.00 & 100.00 & 1.328 & 18.274 & 45.693 & \textbf{0.063} & 2.678 & 9.301 & \textbf{0.247} & 10.23 & 34.84 & 23.562 \\
            & \texttt{MIP+cut} & 200 & 0 & 0 & 200 & 0.00 & 0.00 & 100.00 & \textbf{0.071} & \textbf{0.321} & \textbf{3.859} & \textbf{0.063} & 2.682 & 9.575 & \textbf{0.247} & 10.242 & 35.865 & \textbf{23.558} \\\midrule
           \multirow[c]{3}{*}{\texttt{case118}} & \texttt{SOS1} & 200 & 0 & 47 & 153 & 0.00 & 23.50 & 76.50 & 3.423 & 62.66 & 198.834 & 0.015 & \textbf{0.923} & \textbf{3.595} & 0.043 & \textbf{2.418} & \textbf{8.857} & 36.451 \\
            & \texttt{MIP} & 200 & 0 & 24 & 176 & 0.00 & 12.00 & 88.00 & 2.521 & 70.405 & 197.25 & \textbf{0.0} & 1.085 & 3.841 & \textbf{0.0} & 2.83 & 9.447 & \textbf{36.382} \\
            & \texttt{MIP+cut} & 200 & 0 & 0 & 200 & 0.00 & 0.00 & 100.00 & \textbf{0.8} & \textbf{6.015} & \textbf{23.328} & \textbf{0.0} & 1.34 & 5.441 & \textbf{0.0} & 3.42 & 13.21 & 36.538 \\\midrule
           \multirow[c]{3}{*}{\texttt{case300}} & \texttt{SOS1} & 200 & 0 & 200 & 0 & 0.00 & 100.00 & 0.00 & - & - & - & - & - & - & - & - & - & - \\
            & \texttt{MIP} & 200 & 0 & 200 & 0 & 0.00 & 100.00 & 0.00 & - & - & - & - & - & - & - & - & - & - \\
            & \texttt{MIP+cut} & 200 & 0 & 0 & 200 & 0.00 & 0.00 & 100.00 & \textbf{1.051} & \textbf{6.596} & \textbf{29.134} & \textbf{0.012} & \textbf{0.681} & \textbf{3.582} & \textbf{0.033} & \textbf{1.771} & \textbf{8.979} & \textbf{37.192} \\
            \bottomrule
           \end{tabular}         
    }
    \end{center}
   \vspace{-3mm}
\end{table}
\endgroup

\paragraph{Minimality}
\textbf{DCOPF}: The minimality of CEs for the DCOPF are compared in Figure \ref{fig:dcopf_PND} and Table \ref{tab:dcopf_performance}. We observe that solutions to CE problems for larger networks are generally characterised by smaller $\Delta_{\rm{PND}}^{\dem}$ and $\Delta_{\knn}^{\dem}$. For $\Delta_{\rm{PND}}^{\dem}$, the explanation is intuitively linked to the fact that larger networks accumulate higher peak demands $P_{\max}^{\dem}$ which acts as the normalising factor in \eqref{eq:delta_pnd}. In the case of $\Delta_{\knn}^{\dem}$, this trend indicates that the primitive $k$NN1 counterfactual demand scenario cannot serve as an accurate surrogate for the true, minimal CE, especially for larger test cases. Table~\ref{tab:dcopf_performance} confirms that the solutions obtained from using \texttt{SOS1} can serve as a lower bound for all other methods. While for most cases, the CEs obtained from all 3 methods are of similar magnitudes, \texttt{case30} demonstrates a potential issue with the MIP-based methods. It can be suspected that certain big-$M$ values were chosen too tight, shrinking the dual feasibility space too aggressively and cutting off the optimal CE solution.
\begingroup
\setlength{\tabcolsep}{4pt}
\begin{table}[t]
    \caption{Statistics of solving 200 UC CE problems on various test cases with different methods (best methods in \textbf{boldface})}
    \label{tab:uc_performance}
    \begin{center}
   \vspace{-2mm}
    \resizebox{0.99\linewidth}{!}{
        \begin{tabular}{m{1.1cm}lHHH|l|HHHHll|HHll|HllHH}
            \toprule
            \multirow[c]{2}{*}{\textbf{Test case}} & \multirow[c]{2}{*}{\textbf{Method}} &  \# CEs & \multirow[c]{2}{*}{\textbf{\# inf}} & \multirow[c]{2}{*}{\textbf{\# Timeout}} & \multirow[c]{2}{*}{\textbf{\# Succ.}} & \multirow[c]{2}{*}{\textbf{\% inf}} & \multirow[c]{2}{*}{\textbf{\shortstack[l]{\% time-\\out}}} & \% Success & & \multicolumn{2}{c|}{$T_{\rm{run}}$ [s]} & \multirow[c]{2}{*}{\textbf{\shortstack[l]{\# iter.\\(mean)}}} & & \multicolumn{2}{c|}{$\Delta_{\rm{PND}}^{\dem}$ [\%]} & & \multicolumn{2}{c}{$\Delta_{\knn}^{\dem}$ [\%]} & & \\[0.5mm]
            & & & & & & & & & {min} & \multicolumn{1}{c}{mean} & \multicolumn{1}{c|}{max} & & {min} & \multicolumn{1}{c}{mean} & \multicolumn{1}{c|}{max} & {min} & \multicolumn{1}{c}{mean} & \multicolumn{1}{l}{max} & & \\\midrule
            \multirow[c]{2}{*}{\texttt{case14}} & \texttt{DECOMP} & 200 & 0 & 2 & 198 & 0.00 & 1.00 & 99.00 & 0.119 & 7.269 & 600.0 & \textbf{2.934} & \textbf{0.024} & \textbf{1.743} & \textbf{7.095} & \textbf{0.084} & \textbf{4.623} & \textbf{20.447} & 3.35 & 1.52 \\
            & \texttt{DECOMP+cut} & 200 & 0 & 0 & 200 & 0.00 & 0.00 & 100.00 & \textbf{0.089} & \textbf{1.452} & \textbf{106.227} & 3.42 & \textbf{0.024} & 1.784 & \textbf{7.095} & \textbf{0.084} & 4.696 & \textbf{20.447} & 3.33 & 1.52 \\\midrule
            \multirow[c]{2}{*}{\texttt{case30}} & \texttt{DECOMP} & 200 & 0 & 0 & 200 & 0.00 & 0.00 & 100.00 & 0.142 & 1.543 & 82.854 & \textbf{2.475} & \textbf{0.012} & \textbf{1.58} & \textbf{9.094} & \textbf{0.037} & \textbf{4.449} & \textbf{19.142} & 4.24 & 1.59 \\
            & \texttt{DECOMP+cut} & 200 & 0 & 0 & 200 & 0.00 & 0.00 & 100.00 & \textbf{0.107} & \textbf{0.507} & \textbf{7.312} & \textbf{2.475} & \textbf{0.012} & \textbf{1.58} & \textbf{9.094} & \textbf{0.037} & \textbf{4.449} & \textbf{19.142} & 4.24 & 1.59 \\\midrule
            \multirow[c]{2}{*}{\texttt{\texttt{case57}}} & \texttt{DECOMP} & 200 & 0 & 1 & 199 & 0.00 & 0.50 & 99.50 & 0.202 & 9.582 & 600.0 & \textbf{4.221} & \textbf{0.027} & \textbf{2.397} & \textbf{12.535} & \textbf{0.076} & \textbf{5.37} & \textbf{26.942} & 4.04 & 1.82 \\
            & \texttt{DECOMP+cut} & 200 & 0 & 0 & 200 & 0.00 & 0.00 & 100.00 & \textbf{0.156} & \textbf{2.414} & \textbf{80.508} & 4.37 & \textbf{0.027} & 2.41 & \textbf{12.535} & \textbf{0.076} & 5.387 & \textbf{26.942} & 4.03 & 1.82 \\\midrule
            \multirow[c]{2}{*}{\texttt{case89}} & \texttt{DECOMP} & 200 & 3 & 19 & 181 & 1.50 & 9.50 & 90.50 & \textbf{0.302} & 61.024 & \textbf{600.0} & \textbf{8.26} & \textbf{0.004} & \textbf{1.076} & \textbf{5.251} & \textbf{0.008} & \textbf{3.037} & \textbf{16.06} & 6.55 & 2.02 \\
            & \texttt{DECOMP+cut} & 200 & 4 & 13 & 187 & 2.00 & 6.50 & 93.50 & 0.32 & \textbf{48.734} & \textbf{600.0} & 9.535 & \textbf{0.004} & 1.266 & 13.241 & \textbf{0.008} & 3.437 & 26.436 & 6.57 & 2.02 \\\midrule
            \multirow[c]{2}{*}{\texttt{\texttt{case118}}} & \texttt{DECOMP} & 200 & 0 & 10 & 190 & 0.00 & 5.00 & 95.00 & 1.6 & 39.574 & \textbf{600.0} & 2.684 & \textbf{0.001} & \textbf{0.145} & \textbf{1.648} & \textbf{0.003} & \textbf{0.448} & \textbf{4.397} & 17.90 & 6.36 \\
            & \texttt{DECOMP+cut} & 200 & 0 & 7 & 193 & 0.00 & 3.50 & 96.50 & \textbf{1.258} & \textbf{29.609} & \textbf{600.0} & \textbf{2.58} & \textbf{0.001} & 0.769 & 5.718 & 0.005 & 2.417 & 20.815 & 18.73 & 6.77 \\\midrule
           \multirow[c]{2}{*}{\texttt{\texttt{case300}}} & \texttt{DECOMP} & 200 & 0 & 38 & 162 & 0.00 & 19.00 & 81.00 & 1.373 & 112.033 & \textbf{600.0} & \textbf{1.976} & \textbf{0.001} & \textbf{0.172} & \textbf{2.375} & \textbf{0.002} & \textbf{0.524} & \textbf{5.643} & 16.84 & 7.67 \\
            & \texttt{DECOMP+cut} & 200 & 12 & 17 & 183 & 6.00 & 8.50 & 91.50 & \textbf{1.023} & \textbf{24.769} & \textbf{600.0} & 3.071 & 0.002 & 0.302 & 3.638 & 0.004 & 0.922 & 13.057 & 16.49 & 6.00 \\
            \bottomrule
           \end{tabular}         
    }
    \end{center}
\end{table}
\endgroup

\textbf{Unit Commitment}: 
The results for CE problems on the 6 studied test cases are shown in Figure \ref{fig:uc_PND} and Table \ref{tab:uc_performance}. As for the DCOPF, it is observed that most CE problems can be solved applying small demand variations. The presented statistics for $\Delta_{\rm{PND}}^{\dem}$ indicate that variations are on average less than 2\% of the peak system load for both \texttt{DECOMP} and \texttt{DECOMP+cut}. Again, high values for $\Delta_{\knn}^{\dem}$ occur less often for larger test cases, as shown in Figure \ref{fig:uc_PND}, justifying the usage of algorithmic approaches over the primitive $k$NN1 approach. While the performance worsening (distance increase) introduced by \texttt{DECOMP+cut} compared to \texttt{DECOMP} across all CE problems on all test cases is at most 30.83 times higher (for {\texttt{case300}}), the average increase factor is only 1.126.

\section{Conclusion}\label{sec:conclusion}

This paper develops an explainability framework based on counterfactual explanations (CEs), aimed at explaining dispatch decisions suggested by an optimisation solver for the DC Optimal Power Flow (DCOPF) and the Unit Commitment (UC) problem. We focus on answering why a certain generator was not dispatched as expected by revealing how the electric demand pattern would need to be changed such that a more usual dispatch configuration was restored. To derive CEs for optimisation problems we rely on bilevel problem formulations, for which we compare various solution techniques w.r.t. their computational speed and optimality. For the DCOPF, it is shown that obtaining exact CEs gets increasingly challenging for larger networks with many decision variables. In such cases, a proposed data-driven heuristic can yield significant overall speed improvements with little sacrifices on the minimality of the obtained explanations. Despite the presence of binary variables in the UC problem, CEs can be found within comparable timeframes as for the DCOPF, facilitated by an iterative decomposition algorithm. 

Future work could study the impact of additional security constraints (e.g., N--1) on the explainability of solutions. For the UC problem one could investigate the inclusion of zonal transmission constraints, reserve requirements, or piece-wise production cost curves. Finally, an important step towards a more practical usage of this framework, would be to integrate it into recently emerging natural language processing tools such as \cite{chenOptiChatBridgingOptimization2025}, where users can directly interact with the optimisation software through a dialogue-based interface.

{\appendix

\begin{proof}[Proof of Proposition \ref{proposition}]
    For ease of notation, we define $\Delta\bf{x}=\bf{x}-\bf{x}_{\sz}$ (the element-wise variation of a generic vector $\bf{x}$ from its factual values $\bf{x}_{\sz}$).
    
    We start with the DCOPF case. Every CE problem can be reduced into an inverse optimisation problem \cite{ahujaInverseOptimization2001}, if a full optimal solution is specified, i.e., $\mathcal{X}$ becomes singular. For the DCOPF, with a fully defined solution $\bfp_{\CE}^{\gen}$, the objective function value $\norm{\Delta\bfp^{\dem}}_1$ will be the same for the CE and inverse problem. It is thus sufficient to continue the proof with the latter. The inverse problem is given as
    \begin{subequations}\label{eq:proof}
        \begin{alignat}{3}
            &\min_{\Delta\bfp^{\dem}} \mkern5mu& & \norm{\Delta\bfp^{\dem}}_1\\
            & \text{s.t.} & & \mkern-12mu\sum_{\{n,m\}\in\eSet_n} \mkern-18mu\Delta p^{\line}_{nm} = \mkern-6mu\sum_{g\in\gSet_n}\mkern-6mu\Delta p^{\gen}_g - \mkern-9mu\sum_{d\in\dSet_n}\mkern-6mu\Delta p^{\dem}_d \mkern5mu,\mkern5mu\forall n \in \nSet \mkern12mu({\bm{\eta}}^{\rm{KCL}})\label{seq:proof1}\\
            & & & -\oli{\bfp}^{\line} \leq \bfp^{\line}_{\sz} + \Delta\bfp^{\line} \leq \oli{\bfp}^{\line} \quad(\uli{\bm{\eta}}^{\line},\oli{\bm{\eta}}^{\line}) \label{seq:proof2}\\
            & & & c+\lambda^{\KCL}_{n_g}+\oli{\lambda}^{\gen}\geq0\mkern5mu,\mkern5mu\forall g \in \gSet\quad({\bm{\eta}}^{\gen,\rm{stat}})\label{seq:proof3}\\
            & & & (\ldots)\,,\notag
        \end{alignat}
    \end{subequations}
    where \eqref{seq:proof1}--\eqref{seq:proof2} ensure the feasibility of the flows resulting from $\Delta\bfp^{\dem}$ and \eqref{seq:proof3} is the stationarity constraint w.r.t. $\bfp^{\gen}$. $(\ldots)$ represents the remaining DCOPF and KKT constraints, which are not needed for completion of the proof. Deriving the stationarity condition for $\Delta\bfp^{\dem}$ in problem \eqref{eq:proof}, we can find that $-1\leq\bm{\eta}^{\rm{KCL}}\leq1$.
    Then, from the strong duality theorem applied to \eqref{eq:proof}, we get
    \begin{equation}
        \begin{aligned}
            \norm{\Delta\bfp^{\dem}}_1 = &- \sum_{n\in\nSet}\Bigl(\sum_{g\in\gSet_n}\Delta p^{\gen}_g\Bigr)\underbrace{\eta_n^{\rm{KCL}}}_{\geq -1} \\
            & - \sum_{\{n,m\}\in\eSet}\underbrace{\left(\oli{p}^{\line}_{nm}+{p}^{\line}_{nm,0}\right)\uli{\eta}_{nm}^{\line}}_{\geq0} \\
            & - \sum_{\{n,m\}\in\eSet}\underbrace{\left(\oli{p}^{\line}_{nm}-{p}^{\line}_{nm,0}\right)\oli{\eta}_{nm}^{\line}}_{\geq0} \\
            & - \sum_{g\in\gSet}\underbrace{c_g\eta_g^{\gen,\rm{stat}}}_{\geq0}\\
            &\mkern-23mu\leq\sum_{g\in\gSet}\Delta p^{\gen}_g\leq\bigl\lVert\Delta\bfp^{\gen}\bigr\rVert_1\,.
        \end{aligned}
    \end{equation}

    For the UC problem, the proof is straightforward. By constraint \eqref{seq:uc_p_balance_constr}, power balance holds in each time step, i.e., $\Delta P^{\dem}_t\equiv\sum_{d\in\dSet}\Delta p^{\dem}_{d,t}= \sum_{g\in\gSet}\Delta p^{\gen}_{g,t}\Leftrightarrow |\Delta P^{\dem}_t| = \bigl|\sum_{g\in\gSet}\Delta p^{\gen}_{g,t}\bigr|$. It follows that $\bigl\lVert\Delta\bfp^{\dem}\bigr\rVert_1\leq\bigl\lVert\Delta\bfp^{\gen}\bigr\rVert_1$.
\end{proof}

}


\bibliographystyle{IEEEtran}
\bibliography{IEEEabrv,bib}

\vspace{-33pt}

\begin{IEEEbiographynophoto}
  {Benjamin Fritz} (Student Member, IEEE) received his M.Sc. in Electrical Power Systems from the University of Strathclyde, U.K., in 2025. He is currently a PhD student with the Dyson School of Design Engineering, Imperial College London, U.K. His research interests lie at the intersection of electricity systems and computational decision-making, focussing on areas such as optimisation and machine learning for power system operation, as well as transparency and fairness issues in consumer-centric electricity markets.
\end{IEEEbiographynophoto}
\vspace{-32pt}
\begin{IEEEbiographynophoto}
  {Waqquas Bukhsh} (Senior Member, IEEE) received the B.S. degree in mathematics from COMSATS University Islamabad, Pakistan, in 2008, and the Ph.D. degree in operational research from The University of Edinburgh, U.K., in 2014. He is currently a Senior Lecturer in advanced power systems optimisation with the Department of Electronic and Electrical Engineering, University of Strathclyde, U.K. He works closely with GB's power industry, such as the National Energy System Operator (NESO), to develop decision-support tools for operation and planning. His research interests include model building and numerical optimisation and its application to electricity systems.
\end{IEEEbiographynophoto}



\vfill

\end{document}